\documentclass[sigconf]{aamas} 

\usepackage{graphicx}
\graphicspath{ {./images/} }
\usepackage{amsmath,amsthm,amsfonts,bm,bbm}
\usepackage{tikz}
\usetikzlibrary{automata}
\usetikzlibrary{calc, automata, chains, arrows.meta}
\usetikzlibrary{chains,scopes}
\usetikzlibrary{positioning}

\usepackage{balance} 

\usepackage{xspace}
\usepackage{multirow,subfig}
\usepackage{soul} 

\DeclareMathOperator{\bE}{\mathbb{E}}
\DeclareMathOperator{\bI}{\mathbb{I}}
\DeclareMathOperator{\bR}{\mathbb{R}}

\DeclareMathOperator{\cF}{\mathcal{F}}

\DeclareMathOperator{\malpha}{\overline{\alpha}}

\DeclareMathOperator{\mmu}{\overline{\mu}}
\DeclareMathOperator{\mV}{\overline{V}}
\DeclareMathOperator{\mpi}{\overline{\pi}}

\DeclareMathOperator{\ralpha}{\widehat{\alpha}}

\DeclareMathOperator{\rmu}{\widehat{\mu}}
\DeclareMathOperator{\rV}{\widehat{V}}
\DeclareMathOperator{\rpi}{\widehat{\pi}}

\DeclareMathOperator{\Exp}{\mathbb{E}}
\DeclareMathOperator{\Prob}{\mathbb{P}}

\DeclareMathOperator{\smf}{S^{\mathsf{MF}}}
\DeclareMathOperator{\amf}{A^{\mathsf{MF}}}

\newtheorem{theorem}{Theorem}
\newtheorem{lemma}{Lemma}
\theoremstyle{definition}
\newtheorem{example}{Example}
\newtheorem*{claim*}{Claim}

\newcommand{\mfp}{{\xspace}MFP\xspace}

\setcopyright{ifaamas}
\acmConference[AAMAS '23]{Proc.\@ of the 22nd International Conference
on Autonomous Agents and Multiagent Systems (AAMAS 2023)}{May 29 -- June 2, 2023}
{London, United Kingdom}{A.~Ricci, W.~Yeoh, N.~Agmon, B.~An (eds.)}
\copyrightyear{2023}
\acmYear{2023}
\acmDOI{}
\acmPrice{}
\acmISBN{}

\acmSubmissionID{228}

\title[Indexability is Not Enough for Whittle: Improved, Near-Optimal Algorithms for Restless Bandits]{Indexability is Not Enough for Whittle: Improved, Near-Optimal Algorithms for Restless Bandits}

\author{Abheek Ghosh}
\affiliation{
  \institution{University of Oxford}
  \city{}
  \country{}}
\email{abheek.ghosh@cs.ox.ac.uk}

\author{Dheeraj Nagaraj}
\author{Manish Jain}
\affiliation{
  \institution{Google Research}
  \city{}
  \country{}}
\email{{dheerajnagaraj,manishjn}@google.com}

\author{Milind Tambe}
\affiliation{
  \institution{Google Research}
  \city{} 
  \country{}}
\affiliation{
  \institution{Harvard University}
  \city{}
  \country{}}
\email{milindtambe@google.com}

\begin{abstract}
We study the problem of planning restless multi-armed bandits (RMABs) with multiple actions. This is a popular model for multi-agent systems with applications like multi-channel communication, monitoring and machine maintenance tasks, and healthcare. Whittle index policies, which are based on Lagrangian relaxations, are widely used in these settings due to their simplicity and near-optimality under certain conditions. In this work, we first show that Whittle index policies can fail in simple and practically relevant RMAB settings, \textit{even when} the RMABs are indexable. We discuss why the optimality guarantees fail and why asymptotic optimality may not translate well to practically relevant planning horizons. 

We then propose an alternate planning algorithm based on the mean-field method, which can provably and efficiently obtain near-optimal policies with a large number of arms, without the stringent structural assumptions required by the Whittle index policies. This borrows ideas from existing research with some improvements: our approach is hyper-parameter free, and we provide an improved non-asymptotic analysis which has: (a) no requirement for exogenous hyper-parameters and tighter polynomial dependence on known problem parameters; (b) high probability bounds which show that the reward of the policy is reliable; and (c) matching sub-optimality lower bounds for this algorithm with respect to the number of arms, thus demonstrating the tightness of our bounds. Our extensive experimental analysis shows that the mean-field approach matches or outperforms other baselines. 

\end{abstract}


\newcommand{\BibTeX}{\rm B\kern-.05em{\sc i\kern-.025em b}\kern-.08em\TeX}

\begin{document}


\pagestyle{fancy}
\fancyhead{}


\maketitle 


\section{Introduction}

Multi-Armed Bandits have been extensively studied \cite{lattimore2020bandit,thompson1933likelihood,weber1990index} and widely deployed in online decision-making. Restless multi-armed bandit (RMAB) is a specific instance of this, which deals with optimal allocation of limited resources to multiple agents/arms. Each arm corresponds to a known Markov decision process (MDP), making this a planning problem for multi-agent reinforcement learning. While this is computationally hard~\cite{papadimitriou1999complexity}, there are index policy based approximation algorithms that utilize Lagrangian relaxations of the associated integer programming problems. These policies, called the Whittle index policies~\cite{whittle1988restless}, are simple to implement and are asymptotically optimal under certain structural assumptions on the MDPs. RMABs have been applied across a multitude of domains like multi-channel communication~\cite{hodge2015asymptotic,modi2019transfer,zhao2007myopic,bagheri2015restless,tripathi2019whittle}, monitoring and machine maintenance~\cite{yu2018deadline}, security~\cite{qian2016restless}, and healthcare~\cite{ruiz2020multi,bhattacharya2018restless,lee2019optimal,mate2020collapsing}, including deployments in the field~\cite{mate2022field}. 



However, these index policies require assumptions like homogeneity, indexability, irreducibility, infinite-horizon average reward, and global attractor properties of certain McKean-Vlasov type flows~\cite{weber1990index,hodge2015asymptotic} associated with the Lagrange relaxations to have asymptotic optimality guarantees. These properties can be hard to even verify and practitioners often apply Whittle index policies for their applications without verification.
The general argument for doing so, in the words of~\cite{weber1990index}, is:
\textit{... the evidence so far is that counterexamples to the conjecture (asymptotic optimality of Whittle index policy) are rare and that the degree of sub-optimality is very small.}
In this work, our first contribution is to show that these assumptions cannot be taken for granted even in natural, practically relevant situations by constructing two-action RMABs where Whittle index based policies perform sub-optimally by a large margin.

The most emphasized property required for using the Whittle index policy is indexability. Even Whittle's seminal paper ~\cite{whittle1988restless} that introduced RMABs and proposed the Whittle index policy 
defined the policy only for indexable RMABs. \cite{weber1990index}, among others, show that without indexability, Whittle index policy can be sub-optimal. However, all our examples \textit{are} of indexable RMABs. Therefore, our work further shows that even with indexability, Whittle index policy can be arbitrarily bad compared to the optimal policy for some natural examples. For these examples, even intuitively, the Whittle index policy seems sub-optimal. The examples we provide are well-motivated and are inspired
from practical applications of RMABs, e.g.~\cite{killian2021beyond,mate2022field}.

Our second contribution is to propose the mean-field planning algorithm (\mfp) that works 
provably, reliably, and efficiently with minimal assumptions in settings with a large number of agents $N$. 
\mfp is based on mean-field limits, 
and considers a continuum limit of a system with a large number of similarly behaving agents.
This allows us to remove the combinatorial complexity and obtain analytically tractable objects like differential/difference equations (called the McKean-Vlasov type equations~\cite{weber1990index,hodge2015asymptotic}), albeit with a small approximation error. Our method first maps the problem for any finite $N$ to a continuum limit, where linear programming based methods can be efficiently deployed to solve the problem.
Then, it maps the solution of the continuum limit back to the finite $N$ system to provide us with an approximately optimal solution.

Our third contribution is to provide a 
tight, non-asymptotic convergence of MFP to the optimal policy as a function of $N$. 
In contrast to Whittle index policies, \mfp works without any of the restrictive structural assumptions on the MDPs themselves.
Our algorithmic approach is similar to the algorithm proposed by~\cite{zayas2019asymptotically}, and our analysis shares similarities with the analysis in~\cite{verloop2016asymptotically}. However, our approach does not require precision hyper-parameter tuning of those two approaches and hence is more generic. Our theoretical analysis improves upon the ones in prior works by polynomial factors of problem dependent parameters. 

Our final contribution is to provide an experimental evaluation of \mfp in different settings, especially ones of practical importance. 
Our approach performs as well as or outperforms state-of-the-art index based policies consistently in all the experiments with a large number of agents. 







\section{Related Work}
\paragraph
{\textbf{Mean-Field Limits:}}
Mean field limits are extensively studied in statistical physics and probability theory \cite{weiss1907hypothese,kadanoff2009more,meleard1996asymptotic,budhiraja2012large}, game theory \cite{gueant2011mean,carmona2018probabilistic}, networks and control theory \cite{cammardella2019kullback,cammardella2020kullback,li2010ensemble,chertkov2018ensemble,chen2018distributed,fornasier2014mean,yasodharan2022large}, and reinforcement learning \cite{wang2020breaking,acciaio2019extended,lin2018efficient,yang2017learning} in order to understand the behavior of stochastic dynamical systems with multiple interacting particles or agents. Works like \cite{fornasier2014mean,lacker2017limit,budhiraja2012large} consider convergence of multi-agent continuous time control systems to McKean-Vlasov equations under general conditions. 

\paragraph
{\textbf{Restless Bandits and Approximation Limits:}}
\cite{papadimitriou1994complexity} shows that obtaining the exact solution to the RMAB planning problem is PSPACE-hard.
This led to research on approximation algorithms for the optimal policy whenever the number of arms $N$ is large.
For the case of two action RMABs, the practically effective `Whittle index' policy was proposed in~\cite{whittle1988restless}, and was shown to be asymptotically optimal under certain structural assumptions in~\cite{weber1990index}. \cite{bertsimas2000restless} extends this Lagrangian relaxation based approach to a hierarchy of relaxations culminating in the exact solution. 
Recent works have also extended Langrangian relaxations to consider
more complex MDPs with multiple actions~\cite{killian2021beyond,glazebrook2011general,killian2021q,hodge2015asymptotic}.
In the two-action case, \cite{verloop2016asymptotically} extends the Whittle index to the non-indexable setting, but with other structural assumptions. \cite{brown2020index,hu2017asymptotically} consider Lagrangian relaxations for finite horizon two-action MDPs. While these policies do not require indexability, the guarantees in \cite{brown2020index,hu2017asymptotically} suffer from an exponential dependence on the horizon. In contrast, \cite{zhang2021restless,zhang2022near} consider fluid balance policies where the relaxation allows the resource constraint to hold in expectation, but this does not allow for multiple actions. 


Our main algorithm \mfp~is based on the algorithm proposed in~\cite{zayas2019asymptotically} which does not require indexability or any other structural assumptions on the MDP. The method in~\cite{zayas2019asymptotically} requires additional hyper-parameters to be carefully set in order to ensure that the constraints are satisfied and the relaxation is near-optimal as $N \to \infty$ with an error of $\tilde{O}(\sqrt{N})$. We modify the sub-routine that translates the solution to the mean-field LP to a policy for the RMAB,
and obtain the following four improvements:
(i)
Our method is hyper-parameter free, which is very attractive to the practitioner.
(ii)
Suppose all arms are not identical, but there are $K$ clusters of arms each associated with a different MDP with $|S|$ states. Our error bounds scale as $\sqrt{K|S|N}$ whereas the results in \cite{zayas2019asymptotically} scale as $\sqrt{K^4|S|^4N}$. 
Consequently, when $K|S|$ scales as $o(\sqrt{N}) \cap \omega(N^{1/8})$, our bounds still ensure asymptotically optimal policy unlike the prior work.\footnote{Assuming that the optimal expected reward is $\Theta(N)$, which is usually the case.} 
(iii)
We obtain high probability confidence bounds for the random discounted reward under the mean-field policy, showing that this policy is at most $O(\sqrt{K|S|N})$ away from the optimal \emph{expected} reward with \emph{high probability}.
(iv)  
We obtain matching \sloppy $\Omega(\sqrt{N})$ lower bounds for the sub-optimality, which shows that the analysis is tight with respect to the number of arms $N$.

\section{Notation and Preliminaries}\label{sec:prelim}

We begin by describing a (Multi-Action) Restless Multi-Armed Bandit (RMAB) problem.
Let $T$ denote the time horizon for planning, and $[1:T]$ (shorthanded as $[T]$)  denote the set of positive integers $\{1,2, \ldots,T\}$. In an RMAB, 
each arm $i \in [N]$ at time $t \in [T]$ corresponds to an MDP $(S, A, R_{t,i}, P_{t,i})$, where $S$ is the set of states, and $A$ is the set of actions.
We use $P_{t,i}(s'|s,a)$ to denote arm $i$'s transition probability from state $s$ to $s'$ under action $a$ at time $t$. 
Similarly, let $R_{t,i}: S \times A \rightarrow \bR_{\ge 0}$ and $C_{t,i} : S \times A \rightarrow \bR_{\ge 0}$ denote the reward and the cost, respectively, of playing action $a$ for arm $i$ in state $s$ at time $t$. Notice that the rewards and transitions can change with time.
We assume that there is an action of cost $0$ for each $t$, $i$, and $s$.\footnote{This is without loss of generality for $t$ and $i$: all the costs and budget can always be reduced by the minimum positive cost to get a zero cost action. 
While this condition is \textit{with} loss of generality for $s$, it is not limiting in practice given most RMAB applications have a default zero cost action.}


We require that the actions must satisfy the budget constraint: $\sum_{i \in [N]} C_{t,i}(s_{t,i},a_{t,i}) \le B_t  \in \bR_{\ge 0}$. A policy $\pi = (\pi_t)_{t \in [T]}$ recommends a joint action $(a_{t,i})_{i \in [N]}$ 
as a function of the joint state $(s_{t,i})_{i \in [N]}$ 
i.e., $\pi_t : S^N \rightarrow A^N$ for each $t \in [T]$. Note that, the MDPs of all agents collectively form a single, large MDP with state-space $S^N$. This admits a deterministic optimal policy, which maps the joint state deterministically to a single joint action $(a_{t,i}) \in A^N$ (see \cite{szepesvari2010algorithms}). Therefore, we will work with deterministic policies only. 

\sloppy
The expected discounted reward starting from state $s_1$ under the policy $\pi$ is defined as $V_{\gamma}^\pi(s_1) = \bE\left[\sum_{t=1}^T \gamma^{t-1} \sum_{i \in [N]} R_{t,i}(s_{t,i},\pi_t(s_t)_i) \bigr| s_1\right]$ where the next state is drawn according to $s_{t+1,i} \sim P_{t,i}(\cdot | s_{t,i},\pi_t(s_t)_i)$ independently for $i \in [N]$ when conditioned on $s_t$. $\gamma \in [0,1]$ is the discount factor, the actions are selected according to the policy $\pi$, and $\pi_t(s_t)_i$ is the $i$-th component of $\pi_t(s_t)$. The planner's goal is to maximize the total expected reward.

\subsection{RMABs with Clusters}\label{sec:prelim:rmabClusters}
While known theoretical bounds require homogeneity of the arms (i.e, the associated MDPs are identical), 
different arms behave differently in practically relevant settings. For example, different patients might have different reactions to the same treatment, as observed in healthcare applications~\cite{mate2022field,mate2020collapsing}. A clustering approach, where arms are grouped into clusters is often applied in these scenarios~\cite{mate2022field}.
Each arm belongs to exactly one cluster, and the MDPs associated with all arms in a given cluster are identical.
This setting allows for certain aspects of the model to scale with the number of clusters and not the arms, and 
the model itself can be learned in a statistically efficient way from data (e.g. using collaborative filtering) 
without playing each arm individually multiple times. 

We shall use the notation based on clusters throughout the paper, and we describe this notation in Table~\ref{tab:notation}.
The state of a cluster can be succinctly described as the vector of counts of arms in that given cluster and a given state, due to the homogeneity of arms (i.e, arms in a given state and a given cluster are following identical MDPs and are indistinguishable from each other in terms of reward maximization). Thus, we represent the state of the cluster as
$\rmu_t = (\rmu_{t,i}(s))_{i \in [K], s \in S}$, and the action at time $t$ as $\ralpha_t = (\ralpha_{t,i}(s,a))_{i \in [K], s \in S, a \in A}$.
A few identities are a direct consequence of our definitions: 
\begin{align*}
\sum_{s \in S} \rmu_{t,i}(s) = N_i, \ 
\sum_{a \in A} \ralpha_{t,i}(s,a) = \rmu_{t,i}(s), \ \forall i \in [K], t \in [T], s \in S.
\end{align*}
As an illustrative example, suppose that there are four states $S = \{s_1, s_2, s_3, s_4\}$ and three actions $A = \{a_1, a_2, a_3\}$. Say, cluster $i$ has $100$ arms, and at time $t$, the states $s_1$, $s_2$, and $s_3$ have $30$ arms each and $10$ arms are in the fourth state $s_4$. Then, the succinctly represented state of cluster $i$ is $\rmu_{t,i} = (30, 30, 30, 10)$, where $\rmu_{t,i}(s_1) = \rmu_{t,i}(s_2) = \rmu_{t,i}(s_3) = 30$ and $\rmu_{t,i}(s_4) = 10$. Similarly, for the $30$ arms in cluster $i$ and state $s_1$, if we take action $a_1$ for $15$ arms, $a_2$ for $10$ arms, and $a_3$ for the remaining $5$ arms, then $\ralpha_{t,i}(s_1, \cdot) = (15, 10, 5)$.
\begin{table}[ht]
    \centering
    \begin{tabular}{|p{0.17\columnwidth}|p{0.75\columnwidth}|}
        \hline
        \textbf{Symbol} & \textbf{Meaning} \\ \hline
        $S$ & State space, subscripted as $s$ \\
        $A$ & Action space, subscripted as $a$ \\
        $T$ & Time horizon for the RMAB, subscripted as $t$ \\
        $K$ & Number of clusters, subscripted as $i$ \\
        $P_{t,i}(s'|s,a)$ & Transition probability from $s$ to $s'$, given $a$ at time $t$ for an agent in cluster $i$ for $a \in A, s, s' \in S$ \\
        $R_{t,i}(s,a)$ & Reward in state $s$ for playing action $a$ for an agent in cluster $i$ at time $t$ for $a \in A, s \in S$ \\
        $C_{t,i}(s,a)$ & Cost in state $s$ for playing action $a$ for an agent in cluster $i$ at time $t$ for $a \in A, s \in S$ \\
        $N_i$ & Number of arms in cluster $i$, where all the $N_i$ add up to $N$ arms \\
        $\rmu_{t,i}(s)$ & Number of arms in cluster $i$ in state $s$ at time $t$ \\
        $\ralpha_{t,i}(s,a)$ & Number of times the action $a$ is applied to the set of arms in the state $s$ and cluster $i$ in time $t$ \\
        $\gamma$ & Discount factor, $\gamma \in [0, 1]$ \\
        \hline
    \end{tabular}
    \caption{Notation based on clustered RMABs.}
    \label{tab:notation}
\end{table}
\section{Failure Examples for Index Policies}
\label{sec:fail_case}
We now provide natural examples motivated by applications in mobile health care~\cite{vida,welldoc,armman-mhealth} where the Whittle index policy performs much worse than the optimal policy (and the policy found by the mean-field method). In these domains, healthcare workers reach out to patients through mobile apps or phone calls to get them to adhere to treatment regimens. 
Here, the patients are modeled as arms, with their adherence with the mobile health program being represented as an MDP~\cite{mate2022field,killian2021beyond}. The action space of the healthcare workers are the different modes of communication (including no outreach).
The number of healthcare workers (resource) is limited and not all beneficiaries can be regularly contacted by them. 
Patients with similar MDPs are grouped together in clusters, similar to clustered RMABs we described in Section~\ref{sec:prelim:rmabClusters}. We now describe three specific RMAB examples inspired from this domain, and compute the  Whittle indices for them. 

\begin{example}
\label{ex:hard1}
In this example (see Figure~\ref{fig:hard1}), the healthcare workers make one of two actions per patient---call or no call.
A patient either decides to continue to engage with the program, or not participate further (e.g. delete the application from the mobile device, or block the incoming call number). Here, we model two types of patients: (a) \textit{greedy}: a patient that engages with the program in the first time step but does not continue in the program even after repeated attempts of healthcare workers;
(b) \textit{reliable}: a patient who is careful to engage in the beginning but will remain in the program for long periods if contacted by healthcare workers consistently.
The RMAB model for this scenario 
has two types of arms corresponding to the \textit{reliable} and \textit{greedy} types of patients, three states per arm (\textit{start}, \textit{engaged}, and \textit{dropout}), and two actions (\textit{active}, or call and \textit{passive}, or no call).

\begin{figure}[htpb]
\includegraphics[width=\columnwidth]{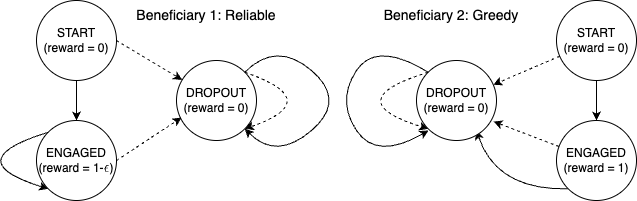}
\caption{Example~\ref{ex:hard1} - Whittle index policy is sub-optimal. (Solid arrows: active actions; dotted arrows: passive actions.)}\label{fig:hard1}
\end{figure}

The behavior of the arms for each type of patient is graphically shown in Figure~\ref{fig:hard1}. The solid lines are transitions for the \textit{active} action ($a=1$), and the dotted lines are for the \textit{passive} action ($a=0$). The nodes represent the $3$ states of \textit{start $s_s$, engaged $s_e$} and \textit{dropout $s_d$} for the patients. For the patient of the \textit{reliable} type, the engaged state gives a reward of $1-\epsilon$ and all other states give a $0$ reward. For the patient of the \textit{greedy} type, the engaged state gives a reward of $1$ and all other states give a $0$ reward. 
The transition probabilities for all transitions for both patient types are given in Table~\ref{tab:example1:transitions}. The parameters for reliable arms are denoted by superscript of $r$ and greedy by $g$. We initialize the RMAB to have $n$ arms of each type (i.e. total arms $N=2n$), as well as the budget $B=n=N/2$ for every time step. All agents start at their respective $s_s$ state.
\end{example}

\begin{table}[htpb]
    \centering
    \begin{tabular}{|c|c|c|}
        \hline 
        & \textbf{Reliable Arm} & \textbf{Greedy Arm} \\ \hline
        \multirow{3}{*}{{Active Action}} & $P^r(s_e|s_s, a=1) = 1$ & $P^g(s_e | s_s, a=1) = 1$ \\
         & $P^r(s_e|s_e, a=1) = 1$ & $P^g(s_d | s_e, a=1) = 1$ \\ 
         & $P^r(s_d|s_d, a=1) = 1$ & $P^g(s_d | s_d, a=1) = 1$ \\
         \hline
         {Passive Action} & \multicolumn{2}{|c|}{$P(s_d | s, a=0) = 1 \ \forall s \in \{s_s, s_e, s_d\}$}\\ 
         \hline
         \multicolumn{3}{|c|}{All other transition probabilities are $0$.} \\
         \hline
    \end{tabular}
    \caption{Transition Probabilities for Example~\ref{ex:hard1}}
    \label{tab:example1:transitions}
\end{table}

The key behavior is that a \textit{reliable} patient moves to the state $s_d$ only after a \textit{passive} action, whereas a \textit{greedy} patient moves to $s_d$ latest by the second time step regardless of the actions. Also, once in $s_d$, an agent stays there forever. So, the greedy arms give rewards for only one time-step, while the reliable arms give rewards for forever if we keep on allocating resources to them. If the discount factor $\gamma$ is not too small, we can intuit that a good policy should allocate resources to the reliable arms. The theorem below shows that the Whittle index does not follow this observation and hence is sub-optimal. We refer to Appendix~\ref{sec:fail_case:app} for its proof.

\begin{theorem}
\label{thm:hardExample1}
The instance described in Example~\ref{ex:hard1} is indexable. The ratio of optimal reward and the Whittle index policy reward is at least $\frac{1-\epsilon}{1-\gamma}$, which goes to $\infty$ as $\gamma \rightarrow 1$ for any $0 < \epsilon < 1$ and $0 < \gamma < 1$.
\end{theorem}


Note that the result in Theorem~\ref{thm:hardExample1} holds for any $n \ge 1$ (and hence $n \rightarrow \infty$), thus making the Whittle index policy asymptotically sub-optimal. Example~\ref{ex:hard1} is indexable (Theorem~\ref{thm:hardExample1}) 
but not homogeneous
(i.e, different arms behave differently with respect to the same action). 
While homogeneity is one of the conditions required by~\cite{weber1990index} for the optimality of the Whittle index policy, \textit{lack of homogeneity is not the reason why the Whittle index policy is sub-optimal} here. 
We show this by transforming this example into a homogenous RMAB 
by combining the reliable and greedy MDPs 
into a single MDP.


\begin{example}
\label{ex:hard2}
As shown in Figure~\ref{fig:hard2}, consider only one type of arm with five states: \textit{reliable-start}, \textit{reliable-engaged}, \textit{greedy-start}, \textit{greedy-engaged}, and \textit{dropout}, which are equivalent to the states in Example~\ref{ex:hard1}. That is, \textit{reliable-start} state behaves like the \textit{start} state of the \textit{reliable} arm, and so on. Similarly, as in Example~\ref{ex:hard1}, we define there to be $2n$ arms ($n$ in each of \textit{reliable-start} and \textit{greedy-start}) and a budget of $n$.\footnote{
We can also modify this MDP so that every arm starts at the same initial state by adding a dummy \textit{start} state while maintaining indexability and homogeneity. This dummy state transitions to either \textit{reliable-start} or \textit{greedy-start} state with probability $0.5$ irrespective of the action played. After one time step, this ensures a similar initial condition as described in Example~\ref{ex:hard2} for large $N$.
}
\end{example}

\begin{figure}[htpb]
\includegraphics[width=0.6\columnwidth]{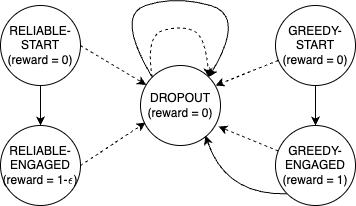}
\caption{Example~\ref{ex:hard2} - Whittle index policy is sub-optimal. (Solid arrows: active actions; dotted arrows: passive actions.)}\label{fig:hard2}
\end{figure}

As in Example~\ref{ex:hard1}, the Whittle index policy prefers the arms in the \textit{greedy-start} state compared to the arms in the \textit{reliable-start} state. The analysis and the results are exactly the same as for Theorem~\ref{thm:hardExample1}: Example~\ref{ex:hard2} is indexable and the Whittle index policy can be arbitrarily worse compared to the optimal policy as $\gamma \rightarrow 1$. 
 


Both these examples show that the Whittle index is sub-optimal for indexable and homogeneous RMABs \textit{even as $N \to \infty$}. Thus, the natural question is: \textit{what about the asymptotic optimality proof of~\cite{weber1990index}}? In addition to these conditions,~\cite{weber1990index} requires that an RMAB satisfy certain `irreducibility' and `global attractor' properties, and proves the optimality of Whittle index for the infinite horizon average reward. However, these conditions often go unverified by designers who implement such policies. 

We now further modify Example~\ref{ex:hard2} to make it irreducible\footnote{A Markov chain is irreducible if there is a non-zero probability to go from any state to any other state after sufficiently many steps.} for every possible single-MDP policy (in addition to homogeneity and indexability). These properties, together with the global attractor property, guarantee that a Whittle index based policy is optimal for the \textit{infinite horizon average reward} problem. However, most environments of practical interest are either finite time horizon or use discounted infinite horizon reward. We show in Example~\ref{ex:hard3} that Whittle index based policies can be sub-optimal in such practical applications, and should be used after careful analysis and experimentation.


\begin{example}
\label{ex:hard3}
In this example, we consider the same \textit{state} and \textit{rewards} as in Example~\ref{ex:hard2}, but 
modify the transition probabilities to ensure irreducibility with respect to any policy. This was not the case previously since, for example, an arm could never transition from $s^g_e$ to $s^r_e$.
Consider the following transition probabilities:
\begin{itemize}
    \item $P(s_{gs},0,s_{ge}) = \eta_{s} = 1 - P(s_{gs},0,s_{d})$,
    \item $P(s_{rs},0,s_{re}) = \eta_{s} = 1 - P(s_{rs},0,s_{d})$,
    \item $P(s_{re},1,s_{d}) = \eta_{r} = 1 - P(s_{re},0,s_{re})$,
    \item $P(s_{d},\cdot,s_{gs}) = P(s_{d},\cdot,s_{rs}) = \eta_d = (1 - P(s_{d},\cdot,s_{d}))/2$,
\end{itemize}
where $s_d$ is the dropout state, $s^g_{s}$ is greedy-start state and so on.
$\eta_s$, $\eta_r$, and $\eta_d$ are positive constants, and are called ergodicity parameters. We set $\eta_s = 0.05$, $\eta_r = 0.1$, $\eta_d = 0.1$, and $\epsilon = 0.01$.
\end{example}

We prove the indexability of Example~\ref{ex:hard3} graphically (see Figure~\ref{fig:hard3:1} and Figure~\ref{fig:hard3:2} in Appendix~\ref{sec:fail_case:app}) by plotting the difference in q-values of active and passive actions as we increase the subsidy for the passive action. Note that the MDP satisfies indexabilitiy since this plot has a negative slope and crosses the x-axis for each state at a unique point, and that the index (x-intercept) of the greedy-start state ($s^g_{s}$) is more than the reliable-start ($s^r_{s}$) state. While we do this for infinite horizon average reward, the same applies for discounted reward and finite horizon reward as well.



In Table~\ref{tab:example3:reward}, we show that the Whittle index is sub-optimal for either discounted reward or finite horizon for Example~\ref{ex:hard3}. The lower bound for the optimal reward is given by the index policy that always prefers $s^r_s$ to 
$s^g_s$.
The difference is also more significant if the ergodicity parameters are made smaller, e.g., if we set $\eta_s = \eta_r = \eta_d = 0.01$, then even with $\gamma = 0.95$, Whittle policy gets a discounted reward of $3.04$ compared to at least $9.32$ for the optimal policy. 

\begin{table}[htpb]
    \centering
    \begin{tabular}{|c|c|c|}
        \hline 
         \textbf{Setting} & \textbf{Whittle} & \textbf{Optimal} \\ \hline
         discounted, $\gamma = 0.95$ & $7.33$ & $\geq 8.65$ \\
         discounted,  $\gamma = 0.8$ & $1.17$ & $\geq 1.86$ \\ 
         finite horizon, $T = 20$ & $7.41$ & $\geq 9.11$ \\
         \hline
    \end{tabular}
    \caption{Estimated rewards for Example~\ref{ex:hard3}.} \label{tab:example3:reward}
\end{table}
For Example~\ref{ex:hard3}, the Whittle index policy is sub-optimal for finite time or discounted reward because the \textit{mixing time} (time it takes to approximately reach the stationary distribution) is large in comparison to the (effective) planning horizon $T$ or $\frac{1}{1-\gamma}$. 
Whittle needs the system to be close to the stationary distribution to start acting near-optimally. 
As the MDP in Example~\ref{ex:hard3} is irreducible for every policy, the mixing time is finite, which ensures optimality of Whittle for non-discounted infinite time reward (given other optimality conditions of~\cite{weber1990index} are also satisfied).
In Appendix~\ref{sec:fail_case:app:mixing}, we show that Whittle's reward goes towards optimality if we either increase the effective horizon or decrease the mixing time for Example~\ref{ex:hard3}.

We have not explicitly tested for the global attractor property in our examples, which states that the McKean-Vlasov differential equation which gives the time evolution of the empirical distribution of the states as $N \to \infty$ under the Whittle index policy converges to a unique fixed point, for every initial condition~\cite{weber1990index}. However, verifying this property is hard in general. In our numerical experiments, the system always converged to the fixed point, 
which makes us believe that our example satisfies the global attractor property.
This further motivates the need for an algorithm whose guarantees does not require these multiple pre-conditions.


\section{The Mean-Field Method}
\label{sec:mf}
 We now propose the mean-field planning (\mfp) method as an approximation technique for finding the optimal policy for an RMAB. 
 The basic idea behind \mfp~is that when $N$ is large, the space of empirical distributions of the states can be approximated closely by the simplex $\Delta(S)$ (up to a quantization error of $O(1/N)$). 
When $N$ is large, a `law of large numbers' type effect shows that the evolution of the empirical measure due to a fixed policy is almost deterministic. 
 Indeed, in \mfp, we use this property to approximate the stochastic process of each cluster's state $\rmu_t = (\rmu_{t,i}(s))_{i \in [K], s \in S}$ by a deterministic process denoted by $\mmu_t = (\mmu_{t,i}(s))_{i \in [K], s \in S}$.

We define the mean-field MDP corresponding to the RMAB instance under consideration below and construct a linear program which solves the mean-field MDP exactly. Our \mfp algorithm, described in Section~\ref{sec:mf:policy}, solves this linear program at each time instant and quantizes the optimal solution in order to compute a feasible action for the RMAB.

\subsection{The Mean-Field MDP}
Given the state of the RMAB and the action taken at $t=1$, $\rmu_1$ and $\ralpha_1$ respectively, the state of the RMAB at $t=2$ is a random variable $\rmu_2 = (\rmu_{2,i}(s'))_{i, s'}$ where $\rmu_{2,i}(s') = \sum_{s, a} \sum_{\ell \in [\ralpha_{1,i}(s,a)]} \bI_{\{x_{\ell}(s,a) = s'\}}$, $x_{\ell}(s,a)$ is a random state in $S$ picked based on the distribution $P_{1,i}(\cdot | s,a)$, 
and $\bI_{\{\ldots\}}$ denotes the indicator function. 
\mfp approximates this \textit{random variable} $\rmu_2$ by the \textit{deterministic} $\mmu_2 = (\mmu_{2,i}(s))_{i, s}$, where 
   $ \mmu_{2,i}(s') = \Exp[\rmu_{2,i}(s')] 
    = \sum_{s \in S, a \in A} P_{1,i}(s'|s,a) \ralpha_{1,i}(s,a).$
In a similar fashion, we can approximate the random variable $\rmu_t$ by the deterministic $\mmu_t$ for each $t \in [T]$. As we shall show later, whenever the number of agents is large, the typical value of $\rmu_{2} - \mmu_2$ is $O(\sqrt{N})$.
Similarly, a mean-field action $\malpha_t = (\malpha_{t,i}(s,a))_{i,s,a}$ roughly corresponds to the expected number of arms in cluster $i$ and state $s$ for which we play the action $a$ at time $t$ and it can be fractional. 

As an illustrative example of the mean-field approximation: Suppose there are two states $S = \{s_1, s_2\}$ and two actions $A = \{a_1, a_2\}$. Let cluster $i$ have $5$ arms, and at $t = 1$, all the arms in cluster $i$ are in state $s_1$, i.e., $\rmu_{1,i} = (5, 0)$. For all $5$ arms in cluster $i$ and state $s_1$, say we play $a_1$, i.e., $\alpha_{1,i}(s_1, \cdot) = (5,0)$. Let's assume $P_{1,i}(\cdot | s_1, a_1) = (0.3, 0.7)$, i.e., an agent in cluster $i$ and state $s_1$ subjected to action $a_1$ moves to state $s_1$ w.p. (with probability) $0.3$ and state $s_2$ w.p. $0.7$. The real-life state at $t = 2$ for cluster $i$ is a random variable $\rmu_{2,i}$, where $\rmu_{2,i} = (\ell, 5 - \ell)$ w.p. $\binom{5}{\ell} (0.3)^{\ell}(0.7)^{5 - \ell}$ for $\ell \in [0:5]$. On the other hand, the mean-field method makes a deterministic approximation of $\rmu_{2,i}$ in the form of the mean-field state $\mmu_{2,i} = \Exp[\rmu_{2,i}] = (1.5, 3.5)$. Notice that $\mmu_{2,i}$ can take fractional values although $\rmu_{2,i}$ can never be fractional. Similarly, $\malpha_{2,i}$ can also be fractional, e.g., $\malpha_{2,i}(s_1, \cdot) = (1, 0.5)$ is a valid mean-field action at time $2$ for this example. 

Recall the original MDP defined in Section~\ref{sec:prelim}. 
The mean-field MDP is a deterministic MDP with state space $\smf = N \Delta([K] \times S) = \{Nx : x \in \Delta([K] \times S)\}$ and action space $\amf \subseteq N \Delta([K] \times S \times A)$, with horizon $T$. 
Given an action $a \in A$, let $P_{t,i}(\cdot|\cdot,a)$ denote the $|S| \times |S|$ transition matrix for cluster $i$ under action $a$ at time $t$.
Using $R_{t,i}(\cdot,a)$ and $C_{t,i}(\cdot,a)$ to denote the $|S|$-dimensional reward and cost vectors, the reward and the cost for playing $\malpha_{t,i}$ for cluster $i$ at time $t$ can be written as $\sum_{a \in A} \malpha_{t,i}(\cdot,a)^{\intercal} R_{t,i}(\cdot,a)$ and $\sum_{a \in A} \malpha_{t,i}(\cdot,a)^{\intercal} C_{t,i}(\cdot,a)$, respectively, where $x^\intercal$ denotes the transpose of a vector $x$. 
An action $\malpha_t$ is feasible at time $t$ if $\sum_{i \in [K]}\sum_{a \in A} \malpha_{t,i}(\cdot,a)^{\intercal} C_{t,i}(\cdot,a) \leq B_{t}$ and $\sum_{a\in A} \malpha_{t,i}(\cdot,a) = \mmu_{t,i}(\cdot)$ for every $i \in [K]$. For all $t \in [T-1]$, we define the mean-field state evolution of state $\mmu_t$ under any feasible action $\malpha_t$ as:
\[
    \mmu_{t+1,i}(\cdot)^\intercal = \sum_{a \in A} \malpha_{t,i}(\cdot,a)^{\intercal} P_{t,i}(\cdot|\cdot,a).
\]

\subsection{The Mean-Field LP}\label{sec:mf:lp}
For the mean-field MDP defined above, we can find the optimal action by solving the following LP. 
This LP gives a mean-field action $\malpha$ computed using the real-life state at time $t=1$, $\rmu_1$, which is known. 
\begin{align}
    \max_{\malpha_t, \mmu_t} \ \  \sum_{t, i, a}  \gamma^{t-1} \malpha_{t,i}(\cdot,a)^{\intercal} & R_{t,i}(\cdot,a) \label{eq:mflp:1}  \\
    \rm{s.t. } \quad \quad \mmu_{1,i} = \rmu_{1,i}, \qquad &\forall i \label{eq:mflp:5} \\ 
    \mmu_{t+1,i}(\cdot)^\intercal = \sum_{a \in A} \malpha_{t,i}(\cdot,a)^{\intercal} P_{t,i}(\cdot|\cdot,a), \qquad &\forall t \in [T-1], \forall i \label{eq:mflp:2} \\
    \sum_{i \in [K]}\sum_{a \in A} \malpha_{t,i}(\cdot,a)^{\intercal} C_{t,i}(\cdot,a) \le B_t, \qquad &\forall t \label{eq:mflp:4} \\
    \sum_{a \in A} \malpha_{t,i}(\cdot,a) = \mmu_{t,i}(\cdot), \qquad & \forall t,i \label{eq:mflp:3} \\ 
    \malpha_{t,i}(\cdot,a) \ge 0, \qquad & \forall t,i,a \label{eq:mflp:6} 
\end{align}
Equation~\eqref{eq:mflp:4} is the budget constraint whereas Equations~\eqref{eq:mflp:3} and~\eqref{eq:mflp:6} guarantee consistency. This LP is similar to the the LP given in~\cite{zayas2019asymptotically} with the hyperparameters $\epsilon$ set to $0$. Next, we describe the mean-field policy computed using this LP - i.e, a policy computed for the real-world RMAB using the mean-field solution $\malpha$. Note that this is different from the randomized policy computation in~\cite{zayas2019asymptotically}.

\subsection{RMAB Policy from the Mean-Field LP}\label{sec:mf:policy}
We are given the initial real-life state $\rmu_1$. For $t = 1, 2, \ldots, T$ execute the steps described below:
\begin{enumerate}
    \item 
    Solve the above LP 
    with the initial state as $\rmu_{t}$ and the time horizon as $(T-t+1)$, corresponding to the MDP considered from times $t$ to $T$. Let's denote the mean-field states and actions computed by solving the LP as $\mmu^{(t)} = (\mmu^{(t)}_{\tau})_{\tau \in [t:T]}$ and $(\malpha^{(t)}_{\tau})_{\tau \in [t:T]}$, respectively. Notice that given an initial state at time $t$, $\mmu^{(t)}_{t} = \rmu_{t}$.
    
    \item 
    Play the action $\malpha^{(t)}_{t}$. Notice that, by construction, $\sum_a \malpha^{(t)}_{t,i}(s,a) = \mmu^{(t)}_{t,i}(s) = \rmu_{t,i}(s)$. But, some of the $\malpha^{(t)}_{t,i}(s,a)$ values may not be integers. So, we take a \texttt{floor} and play the action $a$ for $\lfloor \malpha^{(t)}_{t,i}(s,a) \rfloor$ arms in cluster $i$ and state $s$ at time $t$. For the remaining arms (in cluster $i$ and state $s$ at time $t$), which is equal to $\rmu_{t,i}(s) - \sum_a \lfloor \malpha^{(t)}_{t,i}(s,a) \rfloor$, play the $0$ cost action. This ensures that the `rounded' action stays within the budget $B_t$. 
    
    \item 
    After playing the action at time $t$, we get the realization of the real-life state at time $t+1$, $\rmu_{t+1}$.
\end{enumerate}
 Step 2 introduces a potential \textit{rounding error} in the algorithm. However, in our experiments, we observed that the mean-field actions were generally integral. This is likely because the costs, rewards and budget in our experiments were all integers, and our optimal solution of the resulting LP may be an integral extreme point. Since this might not hold for general RMAB instances, we incorporate and bound the rounding error in our analysis.

In the algorithm above, we solve the LP for each time step, which is an accepted deployment paradigm for practitioners because it allows the algorithm to adjust its policy based on the evolution of the real-life state. 
If a designer prefers to solve the LP only once,
then they may use an alternate algorithm we discuss in Appendix~\ref{sec:mf:alternate:app}.  
While this variant may have lower performance in practice, both the algorithms have similar worst-case error bounds. 
We refer to Appendix~\ref{sec:mf:lowerbound:app} for such a worst-case example.

\section{Analysis: Near-Optimality of MFP}\label{sec:mf:analysis}

In this section, we prove that the mean-field planning (\mfp) method is close to optimal if the number of arms is large compared to $K|S|$. Specifically, we shall show that the reward we get from the real-life process when we play actions computed using the mean-field method is close to the expected reward we can get by playing actions from an optimal real-life policy, with high probability. 

\subsection{Intuition Behind the Proof}
 Lemma~\ref{lm:expected2lp} presents a simple relaxation based argument to show that the optimal reward for the mean-field MDP is always larger than of the optimal expected reward for the RMAB. Therefore, we reduce the proof to showing that the state-evolution given by the application of the mean-field policy to the RMABs (given by $\rmu_t$) `tracks' the mean-field flow given by $\mmu_t$, up-to a small error. That is, we want to bound $\|\rmu_t -\mmu_t^{(1)}\|_1$.

To this end, we decompose the error produced in each time step into two types: (1) 
    Rounding Error, 
    and (2)
    Variance Error.
Rounding error occurs since the actions proposed by the mean-field solution, given by $\malpha^{(t)}_{t}$, need not be integers. As proposed in the algorithm, we quantize these actions so that the budget is not exceeded. We can show by a simple counting argument that this error is at-most $O(K|S||A|)$ and independent of $N$.

The mean-field MDP, computed at time $t$, assumes that $\mmu_t^{(t)}$ transitions to $\mmu_{t+1}^{(t)}$ exactly. However, this is only true in expectation. As we show, the actual empirical state $\rmu_{t+1}$ is such that $\|\mmu_{t+1}^{(t)} - \rmu_{t+1}\|_1 = \tilde{O}(\sqrt{K|S|N})$ with high probability due to the inherent randomness present in the system. We call this the variance error. By bounding the accumulation of these two kinds of errors, we show that the mean-field policy applied to the RMAB performs near optimally.

\subsection{Preliminaries}\label{sec:mf:prelim}
In our analysis, we shall use the Azuma-Hoeffding inequality for bounds on multinomial distributions. See Appendix~\ref{sec:mf:prelim:app} for the proof.

\begin{lemma}[]
\label{lm:multinomial}
Let $X_1, \ldots, X_n$ be independent random vectors in $\{e_1,\dots,e_k\}$ (where $e_i$ denotes the $i$-th standard basis vector in $\bR^k$) 
\begin{enumerate}
    \item $\Exp[|| \sum_{i \in [n]} (X_i - \Exp[X_i]) ||_1] \le \sqrt{kn}$, and
    \item $\Prob[|| \sum_{i \in [n]} (X_i - \Exp[X_i]) ||_1 \ge \epsilon ] \le \delta$ for all $0 < \delta < 1$, where $\epsilon = \sqrt{2\log(2)kn + 2n\log(\tfrac{1}{\delta})}$.
\end{enumerate}
\end{lemma}

To make the analysis concise, let us introduce some additional notation. Let $R_t(\mu_t, \alpha_t) \equiv R_t(\alpha_t) = \sum_{i,s,a} R_{t,i}(s,a) \alpha_{t,i}(s,a)$ denote the total reward at time $t$ for playing an action $\alpha_t$. Let $\mV_{t:T}(\mu_t)$ be the optimal objective value of the mean-field LP for time $t$ to $T$ with $\mu_t$ as the start state at time $t$ (Section~\ref{sec:mf:policy}). Similarly, let $\rV_{t:T}^{\pi}(\mu_t)$ denote the real-life reward for using any RMAB policy $\pi$ for time $[t:T]$ starting with state $\mu_t$. Note that $\rV_{t:T}^{\pi}(\mu_t)$ is random. Let $\rV_{t:T}(\mu_t) = \max_{\pi} \Exp[\rV_{t:T}^{\pi}(\mu_t)]$ be the maximum expected reward achieved by any RMAB policy. 
We shall also use the notation introduced in Section~\ref{sec:mf:policy}: $\mpi$ to denote the mean-field policy;  $\mmu^{(t)}_{\tau}$ and $\malpha^{(t)}_{\tau}$, for $t \in [T]$, $\tau \in [t:T]$, to denote the states and actions as computed by the LPs in Section~\ref{sec:mf:policy}. Let $R_{max} = \max_{t,i,s,a} R_{t,i}(s,a)$.

\subsection{Analysis}\label{sec:mf:actualanalysis}
As mentioned in Section~\ref{sec:prelim}, an RMAB is a large but finite MDP and has an optimal deterministic Markov policy. Therefore, we restrict ourselves to comparing the reward of the mean-field policy to that of deterministic Markov policies. It can also be observed from the description of the mean-field policy in Section~\ref{sec:mf:policy} that it is deterministic and Markov. We shall first focus on the finite horizon non-discounted reward, where $\gamma = 1$ and $T < \infty$ and later extend the analysis to infinite horizon discounted rewards.

\begin{theorem}\label{thm:expecationF}
The mean-field policy $\mpi$ has a non-discounted reward for finite time horizon $T$ that is at most $O(R_{max}T^2\sqrt{K|S|N})$ less than the expected reward of an optimal policy. In particular,
\begin{enumerate}
    \item lower-bound for expected reward: 
    $\Exp[\rV^{\mpi}_{1:T}(\rmu_1)] \ge \rV_{1:T}(\rmu_1) - \frac{T^2 R_{max}}{4} (\sqrt{K|S|N} + 5K|S||A|)$,
    \item high-probability lower-bound for reward: 
    $\rV^{\mpi}_{1:T}(\rmu_1) \ge \rV_{1:T}(\rmu_1) - \frac{T^2 R_{max}}{4} (\sqrt{2 \log(2) K|S|N + 2 N \log(\frac{T}{\delta})} + 5K|S||A|)$
    with probability at least $(1-\delta)$ for every $0 < \delta < 1$.
\end{enumerate}

\end{theorem}

\begin{proof}
We shall prove our result using a sequence of lemmas (Lemmas~\ref{lm:lipsMF} to~\ref{lm:expected2lp}); the proofs are given in the Appendix~\ref{sec:mf:expecationF:app}. 

First, we prove a Lipschitz condition on the objective value of the optimal solution of the mean-field LP. In particular, say $\mu$ and $\mu'$ are two start states such that the $\ell_1$-distance between $\mu$ and $\mu'$ is small, we shall show the difference between the optimal objective value computed by the LP given these start states is also small.
\begin{lemma}\label{lm:lipsMF}
Let $\mu_t$ and $\mu_t'$ be two arbitrary states at time $t$ with $|| \mu_t - \mu_t' ||_1 = \delta$. Then the difference in the optimal objective value of the mean-field LP for time steps $t$ to $T$ starting from $\mu_t$ and $\mu_t'$ is bounded by 
\[
    |V_{t:T}(\mu_t) - V_{t:T}(\mu_t')| \le (T-t+1) R_{max} \delta/2.  
\]
\end{lemma}

In the next lemma, we bound the rounding error introduced in step (2) of the mean-field policy described in Section~\ref{sec:mf:policy}. 
\begin{lemma}\label{lm:rounding}
 $| \rV_{1:T}^{\mpi}(\rmu_1) - \sum_t R_t(\malpha^{(t)}_t)| \le TK|S||A|R_{max}$ almost surely.
\end{lemma}

Next lemma connects the real-life reward for using the mean-field policy to the objective value of the mean-field LP at time $t=1$.
\begin{lemma}\label{lm:real2lp}
We have the following bound: 
$| \sum_t R_t(\malpha^{(t)}_t) - \mV_{1:T}(\mmu_1) | \le \sum_{t \in [2:T]} \frac{(T-t+1) R_{max} || \mmu^{(t)}_t - \mmu^{(t-1)}_t ||_1 }{2}$.
\end{lemma}

Notice that $\mmu^{(t)}_t$ is the state at time $t$ when playing the mean-field policy and $\mmu^{(t-1)}_t$ is the mean-field estimate of the state at time $t$ as estimated by the mean-field LP at time $t-1$. Next two lemmas bound $|| \mmu^{(t)}_t - \mmu^{(t-1)}_t ||_1$.
\begin{lemma}\label{lm:mubound}
We have the following bounds on $|| \mmu^{(t)}_t - \mmu^{(t-1)}_t ||_1$:
\begin{enumerate}
    \item $\Exp[|| \mmu^{(t)}_{t} - \mmu^{(t-1)}_{t} ||_1] \le \sqrt{K|S|N} + K|S||A|$,
    \item $\Prob \bigr[ || \mmu^{(t)}_{t} - \mmu^{(t-1)}_{t} ||_1 \ge \sqrt{2 \log(2) K|S|N + 2 N \log(1/\delta)} + K|S||A| \bigr] \le \delta$, for every $0 < \delta < 1$.
    
\end{enumerate}
    
\end{lemma}

Next we show that the objective value of the mean-field LP at time $t=1$, $\mV_{1:T}(\mmu_1)$, upper bounds the expected reward of any real-life policy.
\begin{lemma}\label{lm:expected2lp}
Let $\rpi$ denote any arbitrary Markov policy. The expected reward of this policy is bounded above by the objective value of the mean-field LP at time $t=1$, i.e., $\Exp[\rV_{1:T}^{\rpi}(\rmu_1)] \le \mV_{1:T}(\mmu_1)$, where $\rmu_1 = \mmu_1$ is an arbitrary start state.
\end{lemma}

We are now ready to prove our theorem using the lemmas given above.
From Lemma~\ref{lm:expected2lp}, we know that $\Exp[\rV_{1:T}^{\rpi}(\rmu_1)] \le \mV_{1:T}(\mmu_1)$ for every Markov policy $\rpi$. So, the optimal expected real-life reward $\rV_{1:T}(\rmu_1) = \max_{\rpi} \Exp[\rV_{1:T}^{\rpi}(\rmu_1)] \le \mV_{1:T}(\mmu_1)$. Further, using Lemma~\ref{lm:rounding} and Lemma~\ref{lm:real2lp}, we have $\mV_{1:T}(\mmu_1) \le \rV^{\mpi}_{1:T}(\rmu_1) +  \sum_{t \in [2:T]} \frac{(T-t+1) R_{max} || \mmu^{(t)}_t - \mmu^{(t-1)}_t ||_1 }{2} + TK|S||A|R_{max}$. Putting together, the optimal real-life reward is bounded above by the reward received by the mean-field policy as follows:
\begin{multline*}
    \rV_{1:T}(\rmu_1) \le \rV^{\mpi}_{1:T}(\rmu_1) + TK|S||A|R_{max}\\
    +  \sum_{t \in [2:T]} \frac{(T-t+1) R_{max} || \mmu^{(t)}_t - \mmu^{(t-1)}_t ||_1 }{2}.
\end{multline*}

\sloppy
Using Lemma~\ref{lm:mubound}, we can lower-bound $\Exp[\rV^{\mpi}_{1:T}(\rmu_1)]$ as
    $\Exp[\rV^{\mpi}_{1:T}(\rmu_1)] \ge \rV_{1:T}(\rmu_1) - \frac{T^2 R_{max}}{4} (5K|S||A| +  \sqrt{K|S|N})$,
and with high probability bound $\rV^{\mpi}_{1:T}(\rmu_1)$ as
    $\Prob\bigr[\rV^{\mpi}_{1:T}(\rmu_1) \ge \rV_{1:T}(\rmu_1) 
    - \frac{T^2 R_{max}}{4} (5K|S||A| + \sqrt{2 \log(2) K|S|N + 2 N \log(T/\delta)})\bigr] \le 1-\delta$, 
for every $0 < \delta < 1$.
\end{proof}

The next theorem proves the result for infinite horizon discounted reward. When the time horizon is infinite, we run the mean-field algorithm for a suitably selected truncated horizon. This proof is a modification of the proof for Theorem~\ref{thm:expecationF}, and is given in Appendix~\ref{sec:mf:expecationIF:app}. 
\begin{theorem}\label{thm:expecationIF}
For the infinite horizon RMAB planning problem with discount factor $\gamma < 1$, the mean-field policy $\mpi$ computed using a suitably truncated horizon has a reward that is at most $O(R_{max}\sqrt{K|S|N}/(1-\gamma)^2)$ less than the expected reward of an optimal policy. In particular,
\begin{enumerate}
    \item lower-bound for expected reward: 
    $\rV_{1:\infty}(\rmu_1) \le \Exp[\rV^{\mpi}_{1:\infty}(\rmu_1)] + \frac{R_{max}((2-\gamma)K|S||A| + \gamma  \sqrt{K|S|N)})}{2(1-\gamma)^2}$ solved using a truncated horizon of $T \ge \frac{2 \sqrt{N}}{ \sqrt{K|S|}} + 1$,
    \item high-probability lower-bound for reward: 
    $\rV_{1:\infty}(\rmu_1) \le \rV^{\mpi}_{1:\infty}(\rmu_1) + \frac{R_{max}((2-\gamma)K|S||A| + \gamma  \sqrt{2 \log(2) K|S|N + 2N \log(N/\delta})}{2(1-\gamma)^2}$
    with probability at least $(1-\delta)$ for every $0 < \delta < 1$ solved using a truncated horizon of $T \ge \frac{\sqrt{2N}}{ \sqrt{\log(2)K|S| + \log(1/\delta)}} + 1$.
\end{enumerate}
\end{theorem}

We complement the error upper bounds of the mean-field algorithm proved in Theorems~\ref{thm:expecationF} and~\ref{thm:expecationIF} with a tight error lower bound with respect to $N$ in Theorem~\ref{thm:mf:lowerbound} given in Appendix~\ref{sec:mf:lowerbound:app}. 



\section{Experiments}\label{sec:exp:main}
We present results of experiments in two simulation environments based on real-world use cases of RMABs: (i) maternity healthcare~\cite{mate2022field}, and (ii) tuberculosis healthcare~\cite{killian2021beyond}.\footnote{Code: \url{https://github.com/google-research/socialgood/tree/mfp/meanfield}.} Both these experiments show the effectiveness of the \mfp algorithm for real-world settings of practical importance for RMABs. 


\begin{figure*}[ht]
  \centering
  \subfloat[ARMMAN Experiments Reward Plot]{
  \includegraphics[width=0.46\textwidth]{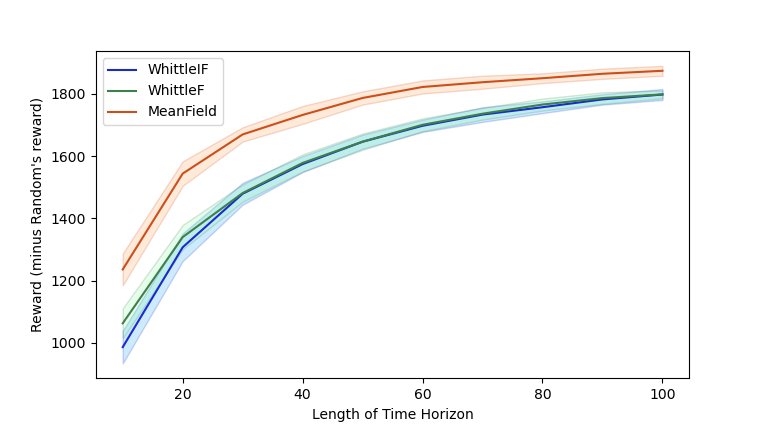}
  \label{fig:armman:1}
  }  
  \subfloat[Tuberculosis Experiment Reward Plot ($1000$ arms)]{
  \includegraphics[width=0.43\textwidth]{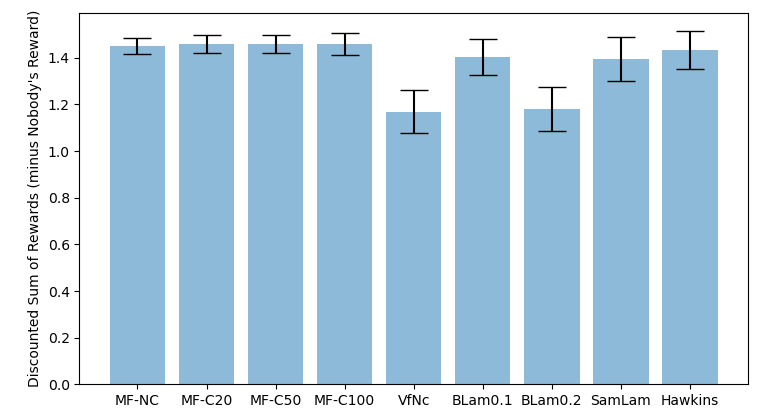}
  \label{fig:tb:1}
  }
\caption{Experimental Results}
\label{fig:armman-tb}
\end{figure*}

\subsection{Mobile Healthcare for Maternal Health}\label{sec:maternity}
This experiment is based on the real-world data of the mMitra program of the healthcare NGO ARMMAN~\cite{armman-mhealth}. Under this program, healthcare workers at ARMMAN make phone calls to enrollees (beneficiaries) of the program to increase engagement and deliver targeted health information. The number of healthcare workers to make phone calls is substantially lower than the number of enrolled beneficiaries, and the healthcare workers have to continually prioritize which beneficiaries to call. \cite{mate2022field} used RMABs to solve this resource allocation problem, modeling each beneficiary as an MDP shown in Figure~\ref{fig:armmanMDP}. \cite{mate2022field} divides the beneficiaries into $40$ clusters, estimates the RMAB parameters (transition probabilities) for each cluster offline, and then uses these parameters to find a resource allocation scheme using the Whittle index policy. 

\begin{figure}[htbp]
\centering
\resizebox{!}{64pt}{%
\begin{tikzpicture}[
    -Triangle, every loop/.append style = {-Triangle},
    start chain=main going right,
    state/.style={circle,minimum size=9mm,draw},
      node distance=14mm,
      font=\scriptsize,
      >=stealth,
      bend angle=48,
      auto
    ]
  \node [state, on chain, fill=black, text=white] (A) {\textbf{\boldmath NE}};
   
  \node[state, on chain, fill=lightgray]  (B) {E};
 
    \path[->,draw, thick, bend right=35]
    (B) edge node[above] {$P(NE|E,a)$} (A);
    \path[thick]
    (B) edge[loop right] node{$P(E|E,a)$} (B);
    
     \path[->,draw, thick, bend right=35]
    (A) edge node[below] {$P(E|NE,a)$} (B);
     \path[thick]
    (A) edge[loop left] node{$P(NE|NE,a)$} (A);

\end{tikzpicture}
}
\caption{MDP model for ARMMAN mMitra program from~\cite{mate2022field}. $NE$ implies that a beneficiary is not engaged with the program, whereas $E$ implies engagement; $a$ is the action.}
\label{fig:armmanMDP}
\end{figure}
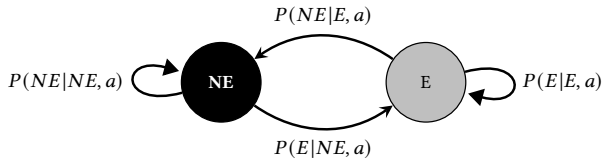

The \mfp algorithm is particularly suitable for the clustered data of ARMMAN.
We communicated with the authors to run our \mfp algorithm on this ARMMAN dataset and we present those results in Figure~\ref{fig:armman:1}. 
We compare the performance of the mean-field policy with the Whittle index policy. We use two versions of the Whittle index policy: (i) the usual infinite-horizon policy; (ii) finite-horizon modification of the Whittle index policy, where the infinite horizon q-values are replaced with finite horizon q-values computed by back-tracking. 
The transition probabilities of $96158$ arms are estimated from the real historical data collected by ARMMAN using the clustering approach of \cite{mate2022field}. 
We use a budget is $1000$, i.e. $1000$ calls per week from healthcare workers to beneficiaries. The discount factor $\gamma$ is set to $0.95$.
The $y$-axis in Figure~\ref{fig:armman:1} is the improvement in the average reward of a policy as compared to the random policy, and the $x$-axis shows the length of the time horizon. 
These results are for $40$ clusters and averaged over $40$ runs. 
See Appendix~\ref{sec:maternity:app} for further details about the experimental environment.

Figure~\ref{fig:armman:1} shows that mean-field performs better than Whittle and has an additional average reward of about $200$. 
\mfp~gives about $20\%$ higher reward (i.e., number of engaged beneficiaries) per \textit{active} action (i.e. phone call) than Whittle. We also observe that mean-field performs better than Whittle for other cluster sizes, given in Appendix~\ref{sec:maternity:app}. Although the worst case time complexity of the mean-field algorithm depends upon the complexity of solving LPs (polynomial but not linear), we see that for our experiments the mean-field algorithm scaled linearly (see Appendix~\ref{sec:maternity:app} for plots). This shows that \mfp may scale well for real-life instances.


\subsection{Tuberculosis Healthcare}\label{sec:tuberculosis}
We test the mean-field algorithm on a rigorous simulation environment (based on a real-world data) developed by \cite{killian2021beyond} motivated by tuberculosis care in India. In this domain, a single healthcare worker manages up to $200$ patients encouraging them to adhere to their $6$-month TB treatment regimen. Each patient is modeled as an individual arm, with state being a tuple of (adherence level, treatment phase, day of treatment). The actions available to the healthcare worker are (i) do nothing (no cost or impact), (ii) call (moderate cost and impact), and (iii) visit (high cost and impact). We refer to \cite{killian2021beyond} for a detailed description of the RMAB. 

We simulate this setting for many parameter combinations. The original simulation framework by \cite{killian2021beyond} did not have any clustering. We ran the mean-field algorithm as: \textbf{MF-NC}, without any clustering ($K = N$); \textbf{MF-C$K$}, with $K$ clusters $K = 20, 50, 100$. We compare our algorithms with the ones in \cite{killian2021beyond}: \textbf{Hawkins}, which is the multi-action extension of Whittle; \textbf{VfNc}, \textbf{SamLam} (called SampleLam in \cite{killian2021beyond}), \textbf{BLam0.1}, and \textbf{BLam0.2}, which are approximations of \textbf{Hawkins} as described in \cite{killian2021beyond}. Figure~\ref{fig:tb:1} plots the reward of these algorithms for $1000$ arms and budget being $10\% = 100$, averaged over $25$ runs (minus the reward of a policy that takes all $0$ cost actions called \textbf{Nobody}). We see that \mfp performs as well as others.
In Appendix~\ref{sec:tuberculosis:app}, we repeat the experiment with smaller number of arms; again, \mfp's performance is as good as the other algorithms.

\textbf{Additional Experiments:} We also simulated another environment proposed in~\cite{killian2021beyond}, which is similar to the examples we show in Section~\ref{sec:fail_case}. 
This simulation environment also has \textit{greedy} and \textit{reliable} type of agents, along with a third agent type named \textit{easy}. 
In these experiments (Appendix~\ref{sec:gre:app}), \mfp significantly outperforms the Whittle index policy. 

\section{Conclusions}
This work challenges the techniques used in practically deployed solutions to important resource allocation problems, and proposes and evaluates an alternative, provably effective method. Specifically, we demonstrate that it is important to read the fine print before applying Whittle index policies for RMABs in sensitive applications like healthcare, by constructing simple, practically relevant and indexable instances of RMABs where Whittle index policies perform poorly. We then propose and analyze the mean field planning (\mfp) algorithm, which is simple, efficient, and hyper-parameter free. This method provably obtains near-optimal policies with a large number of arms $N$ without any structural assumptions of the Whittle index policies. We then demonstrate via extensive experiments that \mfp obtains policies that perform better than the state-of-the-art Lagrangian relaxation based policies---policies that are currently used in practical applications in the healthcare domain.


\begin{acks}
This work was done while AG was doing an internship at Google Research India. We thank Aditya Mate, Shresth Verma, and Jackson Killian for their help with the data and the code for the experiments.
\end{acks}



\bibliographystyle{ACM-Reference-Format} 
\bibliography{ref}


\newpage
\appendix
\section{Appendix for Section~\ref{sec:fail_case}}\label{sec:fail_case:app}
\subsection{Omitted Proofs}\label{sec:fail_case:app:proofs}
\begin{proof}[Proof of Theorem~\ref{thm:hardExample1}]
Let us show that the RMAB given in Example~\ref{ex:hard1} is indexable; we shall get the Whittle indexes as a by-product.

Let $Q^{\lambda}_i(s, a)$ and $V^{\lambda}_i(s)$ be the q-value and v-value, respectively, for an arm of type $i \in \{r,g\}$ in state $s$ and for action $a$, given subsidy for the passive action as $\lambda$ (equivalently, penalty for active action $\lambda$).
Remember that an RMAB is indexable iff the MDP corresponding to each arm is indexable. Further, the MDP corresponding to arm $i$ is indexable iff for every state $s$, there is a value of $\lambda_i(s) \in \bR$ such that $Q^{\lambda}_i(s,0) < Q^{\lambda}_i(s,1)$ for all $\lambda < \lambda_i(s)$ and $Q^{\lambda}_i(s,0) \ge Q^{\lambda}_i(s,1)$ for all $\lambda \ge \lambda_i(s)$. If this indexability condition is satisfied, then this threshold value $\lambda_i(s)$ is called the Whittle index. In other words, we prefer the active action if the subsidy $\lambda$ is less than the Whittle index $\lambda_i(s)$ and prefer the passive action otherwise (given such a $\lambda_i(s)$ exists aka indexability). Also notice that $Q^{\lambda_i(s)}_i(s,0) = Q^{\lambda_i(s)}_i(s,1)$.

Let us now prove indexability and compute the indexes. Let us start with the greedy arms, as they are comparatively simpler to analyze.
\paragraph{\textbf{Arm of Type Greedy}} Let us start with the \textbf{dropout state}. The q-values are:
\begin{equation}\label{eq:thm1:g1}
    Q_g^{\lambda}(s_d,0) = \lambda + \gamma V_g^{\lambda}(s_d); \quad Q_g^{\lambda}(s_d,1) = \gamma V_g^{\lambda}(s_d).
\end{equation}
Therefore, $Q_g^{\lambda}(s_d,0) < Q_g^{\lambda}(s_d,1) \Longleftrightarrow \lambda < 0$. So, the dropout state is indexable with index $\lambda_g(s_d) = 0$. Further, for the value function, if $\lambda < 0$, then $V_g^{\lambda}(s_d) = \max_a Q_g^{\lambda}(s_d, a) = Q_g^{\lambda}(s_d,1) =  \gamma V_g^{\lambda}(s_d) \implies V_g^{\lambda}(s_d) = 0$. If $\lambda \ge 0$, then $V_g^{\lambda}(s_d) = \max_a Q_g^{\lambda}(s_d, a) = Q_g^{\lambda}(s_d, 0) = \lambda + \gamma V_g^{\lambda}(s_d) \implies V_g^{\lambda}(s_d) = \lambda / (1-\gamma)$.
\begin{equation}\label{eq:thm1:g2}
    \text{In summary, } V_g^{\lambda}(s_d) = \begin{cases} 0, &\text{ if $\lambda < 0$,} \\
                                       \frac{\lambda}{1-\gamma}, &\text{ if $\lambda \ge 0$.}               
                         \end{cases}
\end{equation}

Similarly, for the \textbf{engaged state}, the q-values are
\begin{equation}\label{eq:thm1:g3}
    Q_g^{\lambda}(s_e,0) = \lambda + 1 + \gamma V_g^{\lambda}(s_d); \text{ and } \  Q_g^{\lambda}(s_e,1) = 1 + \gamma V_g^{\lambda}(s_d).
\end{equation}
Again, easy to check that $Q_g^{\lambda}(s_e,0) < Q_g^{\lambda}(s_e,1) \Longleftrightarrow \lambda < 0$. So, the engaged state is indexable with index $\lambda_g(s_e) = 0$.

Finally, for the \textbf{start state}, the q-values are
\begin{equation}\label{eq:thm1:g4}
    Q_g^{\lambda}(s_s,0) = \lambda + \gamma V_g^{\lambda}(s_d); \quad Q_g^{\lambda}(s_s,1) = \gamma V_g^{\lambda}(s_e).
\end{equation}
We now need to do a case analysis for $\lambda < 0$ and $\lambda \ge 0$.

\begin{enumerate}
    \item 
First, let us look at the case when $\lambda < 0$. 
This implies $V_g^{\lambda}(s_d) = 0$ from \eqref{eq:thm1:g2}. Also, $V_g^{\lambda}(s_e) = \max_a Q_g^{\lambda}(s_e,a) = Q_g^{\lambda}(s_e,1) = 1 + \gamma V_g^{\lambda}(s_d) = 1$. Therefore, $Q_g^{\lambda}(s_s,1) = \gamma V_g^{\lambda}(s_e) = \gamma > 0 = \gamma V_g^{\lambda}(s_d) > \lambda + \gamma V_g^{\lambda}(s_d) = Q_g^{\lambda}(s_s,0)$.

\item 
Now, let us look at the case when $\lambda \ge 0$. This implies that $V_g^{\lambda}(s_d) = \frac{\lambda}{1-\gamma}$ from \eqref{eq:thm1:g2}. Also, $V_g^{\lambda}(s_e) = Q_g^{\lambda}(s_e,0) = \lambda + 1 + \gamma V_g^{\lambda}(s_d) = 1 + \lambda + \gamma \frac{\lambda}{1-\gamma} = 1 + \frac{\lambda}{1-\gamma}$.
\begin{multline*}
    Q_g^{\lambda}(s_s,1) > Q_g^{\lambda}(s_s,0) \Longleftrightarrow \gamma V_g^{\lambda}(s_e) > \lambda + \gamma V_g^{\lambda}(s_d)  \\ 
    \Longleftrightarrow \gamma + \gamma \frac{\lambda}{1-\gamma} > \lambda + \gamma \frac{\lambda}{1-\gamma} \Longleftrightarrow \lambda < \gamma.
\end{multline*}
\end{enumerate}

Combining the results for $\lambda < 0$ and $\lambda \ge 0$, we get $Q_g^{\lambda}(s_s,1) > Q_g^{\lambda}(s_s,0) \Longleftrightarrow \lambda < \gamma$. So, the start state is indexable and has an index $\lambda_g(s_s) = \gamma$.

So, the MDP corresponding to a greedy arm is indexable.

\paragraph{\textbf{Arm of Type Reliable}} Let us again start with the \textbf{dropout state}. It is easy to check that the dropout state of a reliable arm behaves in exactly the same way as a greedy arm and has the same q-values and v-value. So, we get
\begin{equation}\label{eq:thm1:r1}
    Q_r^{\lambda}(s_d,0) = \lambda + \gamma V_r^{\lambda}(s_d); \quad Q_r^{\lambda}(s_d,1) = \gamma V_r^{\lambda}(s_d).
\end{equation}
\begin{equation}\label{eq:thm1:r2}
    V_r^{\lambda}(s_d) = \begin{cases} 0, &\text{ if $\lambda < 0$,} \\
                                       \frac{\lambda}{1-\gamma}, &\text{ if $\lambda \ge 0$.}               
                         \end{cases}
\end{equation}
Also, the dropout state is indexable with an index of $\lambda_r(s_d) = 0$ as before.

The \textbf{engaged state} of the reliable arm behaves differently from the greedy arm. The q-values have a formula
\begin{equation}\label{eq:thm1:r3}
    Q_r^{\lambda}(s_e,0) = \lambda + (1-\epsilon) + \gamma V_r^{\lambda}(s_d); \quad Q_r^{\lambda}(s_e,1) = (1-\epsilon) + \gamma V_r^{\lambda}(s_e).
\end{equation}
Notice that $V_r^{\lambda}(s_e)$ has a lower bound of
\begin{multline}\label{eq:thm1:r4}
    V_r^{\lambda}(s_e) = \max_a Q_r^{\lambda}(s_e,a) \ge Q_r^{\lambda}(s_e,1) = (1-\epsilon) + \gamma V_r^{\lambda}(s_e) \\
    \implies V_r^{\lambda}(s_e) \ge \frac{1-\epsilon}{1-\gamma}.
\end{multline}
To prove indexability, let us again do a case analysis for $\lambda < 0$ and $\lambda \ge 0$.

\begin{enumerate}
    \item 
If $\lambda < 0$. Then $V_r^{\lambda}(s_d) = 0$ from \eqref{eq:thm1:r2}. We get $Q_r^{\lambda}(s_e,0) = \lambda + (1-\epsilon) + \gamma V_r^{\lambda}(s_d) = \lambda + (1-\epsilon) < (1-\epsilon)$, while $Q_r^{\lambda}(s_e,1) = (1-\epsilon) + \gamma V_r^{\lambda}(s_e) > (1-\epsilon)$. So, $Q_r^{\lambda}(s_e,1) > Q_r^{\lambda}(s_e,0)$.

    \item
If $\lambda \ge 0$. Then $V_r^{\lambda}(s_d) = \frac{\lambda}{1-\gamma}$. We get $Q_r^{\lambda}(s_e,0) = \lambda + (1-\epsilon) + \gamma\frac{\lambda}{1-\gamma} = (1-\epsilon) + \frac{\lambda}{1-\gamma}$. Now, $Q_r^{\lambda}(s_e,1) = (1-\epsilon) + \gamma V_r^{\lambda}(s_e) = (1-\epsilon) + \gamma \max(Q_r^{\lambda}(s_e,0), Q_r^{\lambda}(s_e,1))$, so we have
\begin{itemize}
    \item $Q_r^{\lambda}(s_e,1) > Q_r^{\lambda}(s_e,0) \implies Q_r^{\lambda}(s_e,1) = (1-\epsilon) + \gamma Q_r^{\lambda}(s_e,1) \implies Q_r^{\lambda}(s_e,1) = \frac{1-\epsilon}{1-\gamma}$. As $Q_r^{\lambda}(s_e,0) = (1-\epsilon) + \frac{\lambda}{1-\gamma}$, we get the condition on $\lambda$ as $Q_r^{\lambda}(s_e,1) > Q_r^{\lambda}(s_e,0) \implies  \frac{1-\epsilon}{1-\gamma} > (1-\epsilon) + \frac{\lambda}{1-\gamma} \implies \lambda < \gamma (1-\epsilon)$.
    
    \item $Q_r^{\lambda}(s_e,1) \le Q_r^{\lambda}(s_e,0) \implies Q_r^{\lambda}(s_e,1) = (1-\epsilon) + \gamma Q_r^{\lambda}(s_e,0)$. We get the condition on $\lambda$ as $Q_r^{\lambda}(s_e,1) \le Q_r^{\lambda}(s_e,0) \implies (1-\epsilon) + \gamma Q_r^{\lambda}(s_e,0) \le Q_r^{\lambda}(s_e,0) \implies  Q_r^{\lambda}(s_e,0) \ge \frac{1-\epsilon}{1-\gamma} \implies  (1-\epsilon) + \frac{\lambda}{1-\gamma} \ge \frac{1-\epsilon}{1-\gamma} \implies \lambda \ge \gamma (1-\epsilon)$.
\end{itemize}
\end{enumerate}

Combining the above two results, and also the result for $\lambda < 0$, we get that $Q_r^{\lambda}(s_e,1) > Q_r^{\lambda}(s_e,0) \Longleftrightarrow \lambda < \gamma (1-\epsilon)$. So, this state is indexable with index $\lambda_r(s_e) = \gamma (1-\epsilon)$.

Let us now look at the \textbf{start state} of the reliable arm. We can write the equation for the q-values as
\begin{equation}\label{eq:thm1:r5}
    Q_r^{\lambda}(s_s,0) = \lambda + (1-\epsilon) + \gamma V_r^{\lambda}(s_d); \quad Q_r^{\lambda}(s_s,1) = (1-\epsilon) + \gamma V_r^{\lambda}(s_e).
\end{equation}
Notice that the right hand side of the equations in \eqref{eq:thm1:r5} for the start state match exactly with equations in \eqref{eq:thm1:r3} for the engaged state. In other words, for every $\lambda \in \bR$ and $a \in \{0,1\}$, $Q_r^{\lambda}(s_s,a) = Q_r^{\lambda}(s_e,a)$. So, the start state is also indexable with the same index as the engaged state and equal to $\lambda_r(s_s) = \lambda_r(s_e) = \gamma (1-\epsilon)$.

So, the MDP corresponding to a reliable arm is also indexable. \textit{Overall, the RMAB is thus indexable}.

For any $\epsilon > 0$, at the state state, the Whittle index for a greedy arm $\lambda_g(s_s) = \gamma$ is strictly greater than a reliable arm $\lambda_r(s_s) = \gamma (1-\epsilon)$. Therefore, the Whittle index policy prefers to play the active action for a greedy arm compared to a reliable arm when budget is limited. The discounted sum of rewards for the Whittle index policy is $0 + n\gamma + 0 + \ldots = n\gamma$. On the other hand, we can observe that by always playing the active action for the reliable arms, we get a discounted sum of rewards $0 + n\gamma(1-\epsilon) + n\gamma^2 (1-\epsilon) + n\gamma^3 (1-\epsilon) + \ldots = \frac{n\gamma (1-\epsilon)}{1-\gamma}$. So, the multiplicative loss in reward for the Whittle index policy is $\frac{1-\epsilon}{1-\gamma} \rightarrow \infty$ for $0 < \epsilon < 1$ and $\gamma \rightarrow 1$. Similarly, the additive loss is $\frac{n\gamma (1-\epsilon)}{1-\gamma} - n\gamma \rightarrow \infty$ for $0 < \epsilon < 1$ and $\gamma \rightarrow 1$.
\end{proof}

\begin{proof}[Proof of Indexability of Example~\ref{ex:hard3}]
As described in Example~\ref{ex:hard3}, we prove the indexability of Example 3 graphically. We plot the difference in q-values of active and passive passive actions in Figure~\ref{fig:hard3:1}. Figure~\ref{fig:hard3:2} is a zoomed-in version of Figure~\ref{fig:hard3:1} to show that the negative slope plots cross the $x$-axis at unique points and that the Whittle index (the $x$-intercepts) of the greedy-start state is more than that of the reliable-start state.
\end{proof}
\begin{figure}[ht]
  \centering
  \subfloat[]{
  \includegraphics[width=0.9\linewidth]{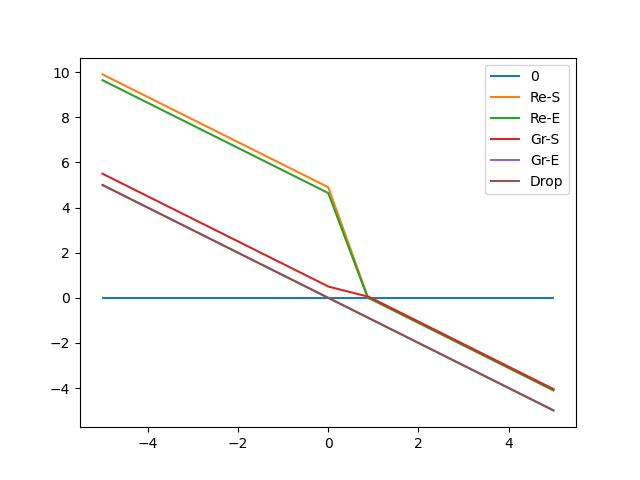}
  \label{fig:hard3:1}
  }\\
  \subfloat[]{
  \includegraphics[width=0.9\linewidth]{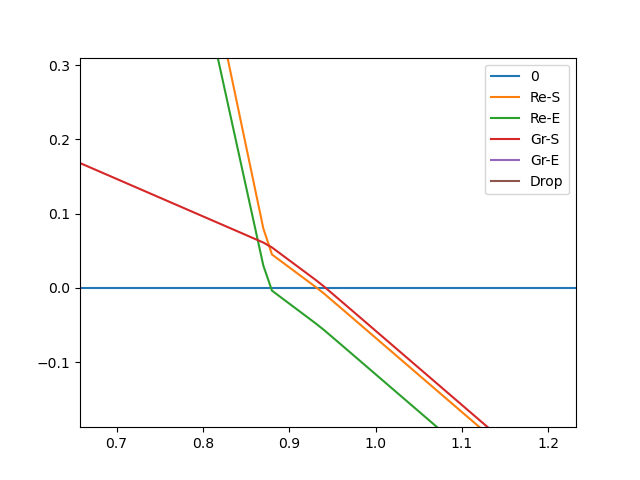}
  \label{fig:hard3:2}
  }
\caption{
Figure~\ref{fig:hard3:1} provides the indexability proof for Example~\ref{ex:hard3} (with $\epsilon = 0.01$, $\eta_s = 0.05$, $\eta_r = 0.1$, $\eta_d = 0.1$ for non-discounted average reward).
Figure~\ref{fig:hard3:2} is a zoomed-in version of Figure~\ref{fig:hard3:1} to see the uniqueness of the $x$-intercept. These figure are best seen in color.
}
\label{fig:hard3:12combined}
\end{figure}





\subsection{Sub-optimality of Whittle with respect to the Mixing Time and the Horizon}\label{sec:fail_case:app:mixing}
In Example~\ref{ex:hard3}, we saw that the Whittle index policy is sub-optimal for finite time or discounted reward even though the MDP satisfied all conditions of \cite{weber1990index}. We also mentioned that the Whittle index policy goes towards optimality as the effective horizon ($T$ for finite time non-discounted reward or $1/(1-\gamma)$ for infinite time discounted reward) increases or the mixing time decreases. Here, we provide the plots for this observation.

We compare the Whittle index policy (\texttt{Whittle}) with an alternate policy (\texttt{Alternate}). \texttt{Alternate} prioritizes the arms in the states in the following preference order: $s^r_e > s^r_s > s^g_s > s_d > s^g_e$. We plot the reward of \texttt{Whittle} and \texttt{Alternate}, and their ratio \texttt{Whittle}/\texttt{Alternate}, with respect to
\begin{enumerate}
    \item The discount factor $\gamma$. As $\gamma$ increases, the effective time horizon for the infinite time discounted reward, $1/(1-\gamma)$, also increases, so Whittle's reward gets close to optimal. Figure~\ref{fig:hard3:gamma}.
    \item The horizon $T$. As the horizon $T$ for finite time non-discounted reward increases, Whittle's reward gets close to optimal. Figure~\ref{fig:hard3:horizon}.
    \item The ergodic parameter $\eta_s$, $\eta_r$, and $\eta_d$ are as defined in Example~\ref{ex:hard3}. We set $\eta_s = \eta_r = \eta_d = \eta$. As $\eta$ increases, the mixing time decreases, so Whittle's reward gets close to optimal. Figure~\ref{fig:hard3:eta}.
\end{enumerate}
When we vary either of these parameters, we keep the other parameters fixed at $\gamma = 0.95$ (except when we have finite horizon, then $\gamma = 1$), $T = \infty$, and $\eta = 0.01$.

\begin{figure*}[ht]
  \centering
  \subfloat[Reward Plot]{
  \includegraphics[width=0.45\textwidth]{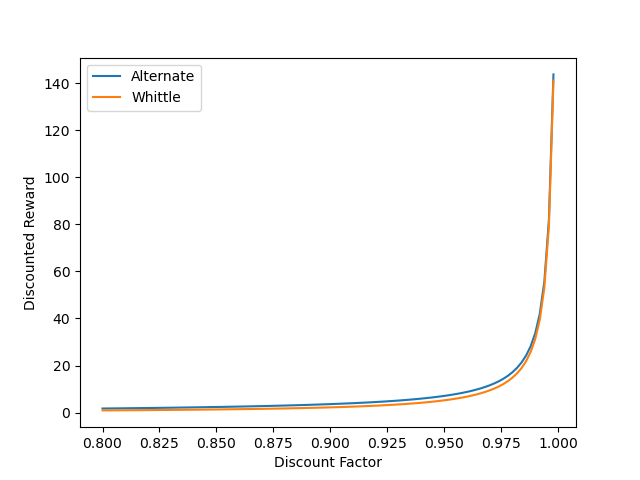}
  }
  \subfloat[Ratio of Rewards Plot]{
  \includegraphics[width=0.45\textwidth]{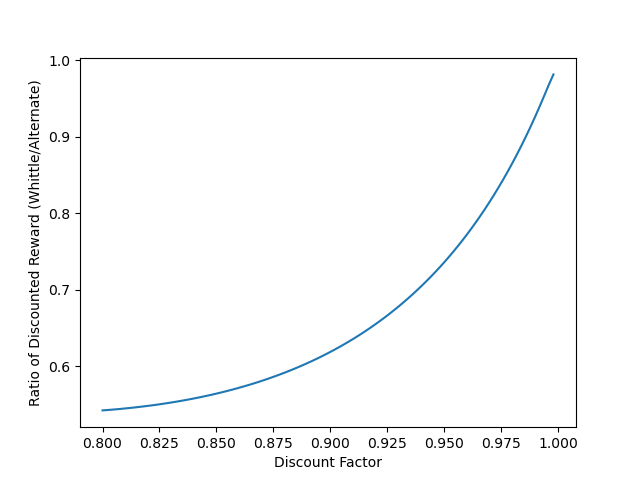}
  }

\caption{Reward vs Discount Factor $\gamma$ (Example~\ref{ex:hard3}).}
\label{fig:hard3:gamma}
\end{figure*}

\begin{figure*}[ht]
  \centering
  \subfloat[Reward Plot]{
  \includegraphics[width=0.45\textwidth]{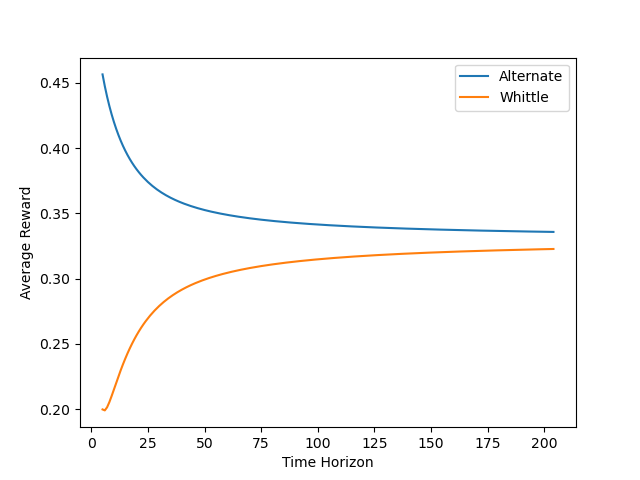}
  }
  \subfloat[Ratio of Rewards Plot]{
  \includegraphics[width=0.45\textwidth]{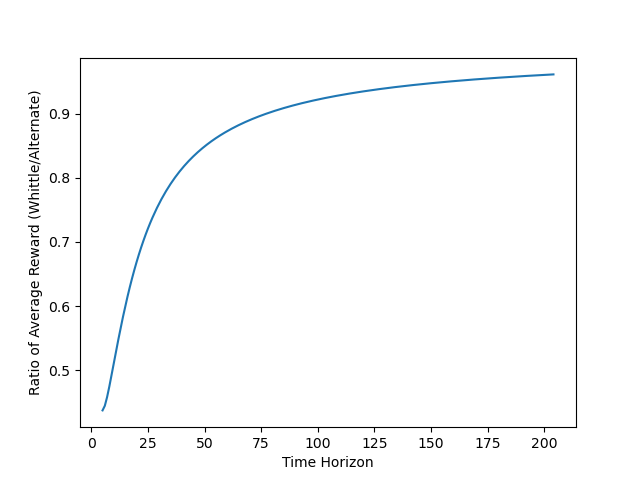}
  }

\caption{Reward vs Finite Time Horizon $T$ (Example~\ref{ex:hard3}).}
\label{fig:hard3:horizon}
\end{figure*}

\begin{figure*}[ht]
  \centering
  \subfloat[Reward Plot]{
  \includegraphics[width=0.45\textwidth]{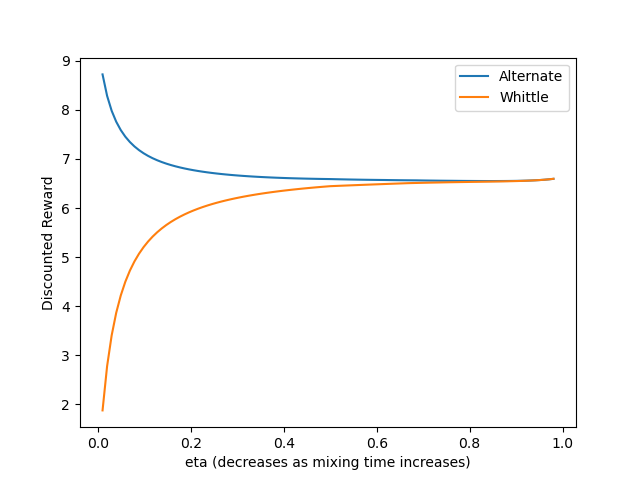}
  }
  \subfloat[Ratio of Rewards Plot]{
  \includegraphics[width=0.45\textwidth]{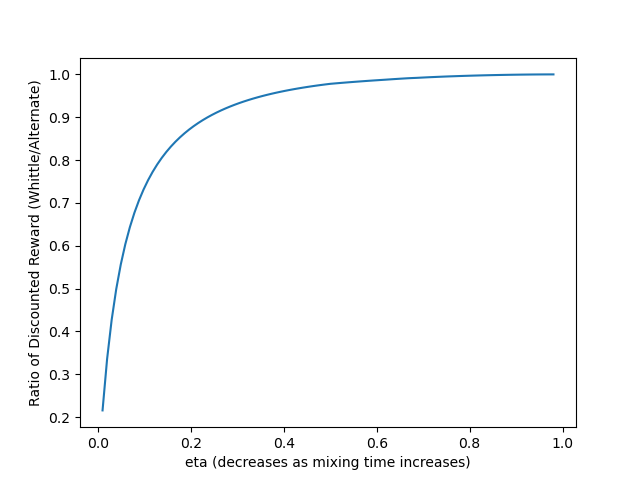}
  }

\caption{Reward vs Ergodic Parameter $\eta$ (Example~\ref{ex:hard3}).}
\label{fig:hard3:eta}
\end{figure*}

\section{Appendix for Section~\ref{sec:mf:analysis}}\label{sec:mf:analysis:app}

\subsection{Appendix for Section~\ref{sec:mf:prelim}}\label{sec:mf:prelim:app}
\begin{lemma}[Hoeffding lemma. \cite{van2014probability} Lemma~3.6]\label{lm:hoeffding}
Let $a \le X \le b$ almost surely for some $a, b \in \bR$. Then $\Exp[e^{\lambda(X-\Exp[X])}] \le e^{\lambda^2(b-a)^2/8}$, i.e., $X$ is $(b-a)^2/4$-subgaussian.
\end{lemma}

\begin{lemma}[Azuma--Hoeffding. \cite{van2014probability} Lemma~3.7]\label{lm:azuma}
Let $(\cF_{k})_{k \le n}$ be any filtration, and let $X_1, \ldots, X_n$ be random variables that satisfy the following properties for $k = 1, \ldots, n$:
\begin{enumerate}
    \item Martingale difference property: $X_k$ is $\cF_k$-measurable and $\Exp[X_k \mid \cF_{k-1}] = 0$.
    \item Conditional subgaussian property: $\Exp[e^{\lambda X_k} \mid \cF_{k-1}] \le e^{\lambda^2 \sigma_k^2 / 2}$ almost surely.
\end{enumerate}
Then the sum $\sum_{k=1}^n X_k$ is subgaussian with variance proxy $\sum_{k=1}^n \sigma_k^2$, which can be used to obtain, for every $\epsilon \ge 0$, the tail bound
\[
    \Prob[\sum_{k=1}^n X_k \ge \epsilon] \le \exp{\left( - \frac{ \epsilon^2}{2 \sum_{k=1}^n \sigma_k^2} \right)}.
\]
The same probability bound holds for $\Prob[\sum_{k=1}^n X_k \le -\epsilon]$ as well (replace $X_k$ with $-X_k$).
\end{lemma}

\begin{proof}[Proof of Lemma~\ref{lm:multinomial}]
Let $p_{ij} = \Prob[X_i = e_j]$, and therefore, $\Exp[\langle X_i,e_j\rangle] = p_{ij}$.

Therefore,
\begin{align}
    &\Exp[|| \sum_{i \in [n]} (X_i - \Exp[X_i]) ||_1] = \Exp[\sum_{j=1}^{k}\bigr| \sum_{i \in [n]} \langle e_j,(X_i - \Exp[X_i])\rangle \bigr|] \nonumber \\
    &\leq \sum_{j=1}^{k}\sqrt{\Exp[\langle\sum_{i\in [n]}X_i - \Exp[X_i],e_j\rangle^2]} \\
    &= \sum_{j=1}^{k}\sqrt{  \sum_{i\in [n]} \Exp[(\langle X_{i},e_j\rangle - p_{ij})^2]}
    \nonumber \\&= \sum_{j=1}^{k}\sqrt{  \sum_{i\in [n]}p_{ij}(1-p_{ij})} \leq \sum_{j=1}^{k}\sqrt{  \sum_{i\in [n]}p_{ij}}.
\end{align}

In the first step, we have used the definition of $\ell_1$ norm. 
In the second step, we have used Jensen's inequality to bound the first moment using the second moment. 
In the third step, we have computed the variance of the Bernoulli random variable $\langle X_{i},e_j\rangle$. 
Note that $\sum_{j\in [k]}\sum_{i\in [n]} p_{ij} = n$. Using the fact that $\sum_{j=1}^{k}\sqrt{n_j}$ subject to the constraint $\sum_j n_j = n$ is maximized when $n_j = \frac{n}{k}$, we conclude from the display above that: 
\[
    \Exp[|| \sum_{i \in [n]} (X_i - \Exp[X_i]) ||_1] \leq \sqrt{kn}.
\]

Note that since $X_{i}$ is almost surely a standard basis vector, therefore we must have $\sum_j p_{ij} = 1$ and $ \sum_j \langle X_i,e_j \rangle = 1$ almost surely. Therefore, given any arbitrary $\alpha := (\alpha_1,\dots, \alpha_k) \in \{-1,1\}^k$,
$a_i(\alpha) \leq \sum_{j \in [k]} (\langle X_i,e_j\rangle - p_{ij})\alpha_j \leq a_i(\alpha)+2 $ for some $a \in \mathbb{R}$ almost surely. Therefore, applying 
Lemmas~\ref{lm:hoeffding} and~\ref{lm:azuma} to the random variables $Y_i(\alpha) := \sum_{j \in [k]} (\langle X_i,e_j\rangle - p_{ij})\alpha_j$ for $i \in [n]$, we conclude:

$$\Prob[\sum_{i \in [n]} Y_i(\alpha) \geq \epsilon] \leq \exp(-\tfrac{\epsilon^2}{2n})$$

Now, observe that $$|| \sum_{i \in [n]} (X_i - \Exp[X_i]) ||_1 = \sup_{\alpha \in \{-1,1\}^k} \sum_{i\in [n]}Y_i(\alpha)  $$
Applying the union bound, we conclude:

\begin{align}
    &\Prob[|| \sum_{i \in [n]} (X_i - \Exp[X_i]) ||_1 \ge \epsilon] = \Prob[\cup_{\alpha \in \{-1,1\}^k}\{\sum_{i\in [n]}Y_i(\alpha) \ge \epsilon\}] \nonumber \\
    &\leq \sum_{\alpha \in \{-1,1\}^k} \Prob[\{\sum_{i\in [n]}Y_i(\alpha) \ge \epsilon\}] \nonumber \\
    &\leq 2^{k}\exp(-\tfrac{\epsilon^2}{2n})
\end{align}

Taking $\epsilon = \sqrt{2\log(2)kn + 2n\log(\tfrac{1}{\delta})}$ above, concludes the statement of Lemma~\ref{lm:multinomial}. 
\end{proof}

\subsection{Sub-Optimality Upper Bound for Finite Horizon}
\label{sec:mf:expecationF:app}
\begin{proof}[Proof of Lemma~\ref{lm:lipsMF}]
Let us prove Lemma~\ref{lm:lipsMF} for $t=1$ given $\mu_1$ and $\mu_1'$. The proof trivially generalizes to all $t \in [T]$.

The lemma statement, which we need to prove, says that $|| \mu_1 - \mu_1' ||_1 = \delta \implies |V_{1:T}(\mu_1) - V_{1:T}(\mu_1')| \le R_{max} T \delta/2$. Notice that $|| \mu_1 - \mu_1' ||_1 = \sum_{i,s} |\mu_{1,i}(s) - \mu_{1,i}'(s)| = \sum_{i,s} \max(0,\mu_{1,i}(s) - \mu_{1,i}'(s)) + \sum_{i,s} \max(0, \mu_{1,i}'(s) - \mu_{1,i}(s))$. As $\mu_1$ and $\mu_1'$ both sum up to the same value, the number of arms $N$, i.e., $\sum_{i,s} \mu_{1,i}(s) = \sum_{i,s} \mu_{1,i}'(s) = N$, we can see that $\sum_{i,s} \max(0, \mu_{1,i}(s) - \mu_{1,i}'(s)) = \sum_{i,s} \max(0, \mu_{1,i}'(s) - \mu_{1,i}(s)) = \delta/2$. 

Notice that it suffices to prove $\sum_{i,s} \max(0, \mu_{1,i}(s) - \mu_{1,i}'(s)) = \delta/2 \implies V_{1:T}(\mu_1) - V_{1:T}(\mu_1') \le T R_{max} \delta/2 \Longleftrightarrow V_{1:T}(\mu_1') \ge V_{1:T}(\mu_1) - T R_{max} \delta/2$. Because if we swap the variables $\mu_1$ and $\mu_1'$, by symmetry, we get other side of the required inequality: $\sum_{i,s} \max(0, \mu_{1,i}'(s) - \mu_{1,i}(s)) = \delta/2 \implies V_{1:T}(\mu_1') - V_{1:T}(\mu_1) \le T R_{max} \delta/2$

Let $\mu = (\mu_t)_{t \in [T]}$ and $\alpha = (\alpha_t)_{t \in [T]}$ be the optimal solution of the LP given start state $\mu_1$. Notice that by definition $V_{1:T}(\mu_1) = \sum_t R_t(\alpha_t)$. We shall construct a feasible solution $(\mu', \alpha')$ when the start state is $\mu_1'$ that has objective value $\sum_t R_t(\alpha_t') \ge V_{1:T}(\mu_1) - T R_{max} \delta/2$, which implies that $V_{1:T}(\mu_1') \ge \sum_t R_t(\alpha_t') \ge V_{1:T}(\mu_1) - T R_{max} \delta/2$, as required.

Let us define the feasible solution $(\mu', \alpha')$ by decomposing it into two parts: $(\mu^{A}, \alpha^{A})$ and $(\mu^{B}, \alpha^{B})$, where $\mu' = \mu^{A} + \mu^{B}$ and $\alpha' = \alpha^{A} + \alpha^{B}$. 
Let us first focus on $(\mu^{A}, \alpha^{A})$. Let $\mu^{A}_1 = \min(\mu_1, \mu_1')$ (coordinate wise). Observe that by definition, $\mu^{A}_1 \le_{all} \mu_1$ and $\sum_{i,s} (\mu_{1,i}(s) - \mu^A_{1,i}(s)) = \delta/2$ because:
\begin{multline*}
    \mu^{A}_{1,i}(s) = \min(\mu_{1,i}(s), \mu_{1,i}'(s)) \\
    = \mu_{1,i}(s) + \min(0, \mu_{1,i}'(s) - \mu_{1,i}(s)) \\
    = \mu_{1,i}(s) - \max(0, \mu_{1,i}(s) - \mu_{1,i}'(s)) \\
    \implies \sum_{i,s}( \mu_{1,i}(s) - \mu^{A}_{1,i}(s) ) = \sum_{i,s} \max(0, \mu_{1,i}(s) - \mu_{1,i}'(s)) = \delta/2.
\end{multline*}

We define $\alpha^{A}_1$ as $\alpha^{A}_{1,i}(s,a) = \alpha_{1,i}(s,a) \frac{\mu^{A}_{1,i}(s)}{\mu_{1,i}(s)}$. As $\mu^{A}_1 \le_{all} \mu_1$ therefore $\alpha^{A}_1 \le_{all} \alpha_1$. So, at $t=1$, the total cost for playing $\alpha^{A}_1$ is at most the cost for playing $\alpha_1$, which is at most $B_1$. Now, for the reward at time $t=1$, as $\sum_{i,s} (\mu_{1,i}(s) - \mu^A_{1,i}(s)) = \delta/2$, therefore $R_1(\alpha^A_1) \ge R_1(\alpha_1) - R_{max}\delta/2$ as shown below
\begin{align*}
    R_1(\alpha^A_1) &= \sum_{i,s,a} R_{1,i}(s,a) \alpha_{1,i}^A(s,a) \\ 
    &= \sum_{i,s,a} R_{1,i}(s,a) \alpha_{1,i}(s,a) \frac{\mu^{A}_{1,i}(s)}{\mu_{1,i}(s)} \\
    &= \sum_{i,s,a} R_{1,i}(s,a) \alpha_{1,i}(s,a) \frac{\mu_{1,i}(s) - (\mu_{1,i}(s) - \mu^A_{1,i}(s))}{\mu_{1,i}(s)} \\
    &= \sum_{i,s,a} R_{1,i}(s,a) \alpha_{1,i}(s,a) \\
    &\qquad- \sum_{i,s,a} R_{1,i}(s,a) \alpha_{1,i}(s,a) \frac{\mu_{1,i}(s) - \mu^A_{1,i}(s)}{\mu_{1,i}(s)} \\
    &\ge R_1(\alpha_1) - \sum_{i,s,a} R_{max} \alpha_{1,i}(s,a) \frac{\mu_{1,i}(s) - \mu^A_{1,i}(s)}{\mu_{1,i}(s)} \\
    &= R_1(\alpha_1) - \sum_{i,s} R_{max} (\mu_{1,i}(s) - \mu^A_{1,i}(s)) \\
    &= R_1(\alpha_1) - R_{max} \delta/2.
\end{align*}

Let us now look at time $t=2$. The mean-field LP computes $\mu_2$ from $\alpha_1$, and $\mu^{A}_2$ from $\alpha^{A}_1$, using the linear equation \eqref{eq:mflp:2}. As $\alpha^{A}_1 \le_{all} \alpha_1$, it is easy to check from \eqref{eq:mflp:2} that $\mu^A_2 \le_{all} \mu_2$. Also, as $\sum_{i,s} (\mu_{1,i}(s) - \mu^A_{1,i}(s)) = \delta/2$, using \eqref{eq:mflp:2}, we also have $\sum_{i,s} (\mu_{2,i}(s) - \mu^A_{2,i}(s)) = \delta/2$ as shown below
\begin{multline*}
    \sum_{i,s'} (\mu_{2,i}(s') - \mu^A_{2,i}(s')) = \sum_{i,s,a,s'} (\alpha_{1,i}(s,a) - \alpha^A_{1,i}(s,a)) P_{1,i}(s'|s,a) \\
    = \sum_{i,s,a} (\alpha_{1,i}(s,a) - \alpha^A_{1,i}(s,a)) = \sum_{i,s} (\mu_{1,i}(s,a) - \mu^A_{1,i}(s)) = \delta/2.
\end{multline*}
As we did for time $t=1$, if we define $\alpha^{A}_{2,i}(s,a) = \alpha_{2,i}(s,a) \frac{\mu^{A}_{2,i}(s)}{\mu_{2,i}(s)}$, using $\mu^A_2 \le_{all} \mu_2$ we can show that the budget constraint is satisfied, and using $\sum_{i,s} (\mu_{2,i}(s) - \mu^A_{2,i}(s)) = \delta/2$ we can show that the reward lower bound of $R_2(\alpha^A_2) \ge R_2(\alpha_2) - R_{max}\delta/2$ is satisfied. Repeating this process for all time steps, we can show that for every $t$, the budget constraint is satisfied $\sum_{i,s,a} \alpha^A_{t,i}(s,a) C_{t,i}(s,a) \le \sum_{i,s,a} \alpha_{t,i}(s,a) C_{t,i}(s,a) \le B_t$, and the reward is lower bounded $R_t(\alpha^A_t) \ge R_t(\alpha_t) - R_{max}\delta/2$. So, the total reward for $\alpha^A$ is also lower bounded as: $\sum_t R_t(\alpha^A_t) \ge \sum_t R_t(\alpha_t) - T R_{max} \delta/2$.

Let us now look at $(\mu^B, \alpha^B)$. Let $\mu^B_1 = \mu_1 - \mu^A_1 = \mu_1 - \min(\mu_1, \mu_1')$. As per our assumptions, for every time $t$, cluster $i$, and state $s$, there is a $0$ cost action; let us denote it by $a^0_{t,i,s}$. We set $\alpha^B_{1,i}(s,a^0_{t,i,s}) = \mu_{t,i}(s)$ and $\alpha^B_{1,i}(s,a) = 0$ for $a \neq a^0_{t,i,s}$. Then, we let the system evolve as per \eqref{eq:mflp:2} to $\mu^B_2$, play an $\alpha^B_2$ similar to $\alpha^B_1$ by selecting only the $0$ cost actions, and repeat this process. The cost incurred by $\alpha^B_t$ is $0$ for each $t$ (and the reward in non-negative as per our assumption on the reward function). By setting $\mu' = \mu^{A} + \mu^{B}$ and $\alpha' = \alpha^{A} + \alpha^{B}$ we get our required result of $V_{1:T}(\mu_1') \ge V_{1:T}(\mu_1) - T R_{max} \delta/2$, which completes the proof.
\end{proof}

\begin{proof}[Proof of Lemma~\ref{lm:rounding}]
In step (2) of the mean-field policy (Section~\ref{sec:mf:policy}), we play the action $a$ for at least $\lfloor \malpha^{(t)}_{t,i}(s,a) \rfloor \ge \malpha^{(t)}_{t,i}(s,a) - 1$ arms in cluster $i$ and state $s$ at time $t$. So, the total reward for the mean-field policy is at least $\sum_{t,i,s,a} (\malpha^{(t)}_{t,i}(s,a) - 1) R_{t,i}(s,a) \ge \sum_t R_t(\malpha^{(t)}_t) - TK|S||A|R_{max}$. With a similar argument we can also show that the real-life reward is $\le \sum_{t,i,s,a} \lfloor \malpha^{(t)}_{t,i}(s,a) \rfloor R_{t,i}(s,a) + \sum_{t,i,s} |A| R_{max}  \le \sum_t R_t(\malpha^{(t)}_t) + TK|S||A|R_{max}$.
\end{proof}

\begin{proof}[Proof of Lemma~\ref{lm:real2lp}]
$\sum_t R_t(\malpha^{(t)}_t)$ can be written as
\begin{align*}
    \sum_{t} R_t(\malpha^{(t)}_t) &= \sum_{t \in [T]} \left( R_t(\malpha^{(t)}_t) + \sum_{\tau = t+1}^{T} R_{\tau}(\malpha^{(t)}_{\tau}) - \sum_{\tau = t+1}^{T} R_{\tau}(\malpha^{(t)}_{\tau}) \right) \\
    &= \sum_{t \in [T]} \left( \sum_{\tau = t}^{T} R_{\tau}(\malpha^{(t)}_{\tau}) - \sum_{\tau = t+1}^{T} R_{\tau}(\malpha^{(t)}_{\tau}) \right) 
\end{align*}
Notice that $\sum_{\tau = t}^{T} R_{\tau}(\malpha^{(t)}_{\tau})$ is the objective value of the mean-field LP given start state $\mmu^{(t)}_t$ for the time steps $[t:T]$, which is denoted by $\mV_{t:T}(\mmu^{(t)}_t)$. Similarly, $\sum_{\tau = t+1}^{T} R_{\tau}(\malpha^{(t)}_{\tau})$ the objective value of the mean-field LP given start state $\mmu^{(t)}_{t+1}$ for the time steps $[t+1:T]$, which is denoted by $\mV_{t+1:T}(\mmu^{(t)}_{t+1})$. Substituting this, we get 
\begin{align*}
    \sum_{t} R_t(\malpha^{(t)}_t) &= \sum_{t \in [T-1]} \left( \mV_{t:T}(\mmu^{(t)}_t) - \mV_{t+1:T}(\mmu^{(t)}_{t+1}) \right) + \mV_{T:T}(\mmu^{(T)}_T) \\
    &= \mV_{1:T}(\mmu^{(1)}_1) + \sum_{t \in [2:T]} \left( \mV_{t:T}(\mmu^{(t)}_t) - \mV_{t:T}(\mmu^{(t-1)}_{t})  \right).
\end{align*}
Using Lemma~\ref{lm:lipsMF}, we get 
\begin{multline*}
    | \sum_{t} R_t(\malpha^{(t)}_t) - \mV_{1:T}(\mmu^{(1)}_1) | \le \sum_{t \in [2:T]} \left| \mV_{t:T}(\mmu^{(t)}_t) - \mV_{t:T}(\mmu^{(t-1)}_{t})  \right| \\
    \le \sum_{t \in [2:T]} \frac{(T-t+1) R_{max} || \mmu^{(t)}_t - \mmu^{(t-1)}_t ||_1 }{2} .
\end{multline*}
\end{proof}

\begin{proof}[Proof of Lemma~\ref{lm:mubound}]
Notice that $\mmu^{(t)}_t$ is the state at time $t$ when playing the mean-field policy and $\mmu^{(t-1)}_t$ is the mean-field estimate of the state at time $t$ as estimated by the mean-field LP at time $t-1$. There are two reasons for $\mmu^{(t)}_t$ to differ from $\mmu^{(t-1)}_t$: (i) bias due to rounding of the actions at time $t-1$; (ii) variance due to the randomness in transitions. As we play the action prescribed by $\malpha^{(t-1)}_{t-1}$ for $\sum_{i,s,a}\lfloor \malpha^{(t-1)}_{t-1,i}(s,a) \rfloor \ge N - K|S||A|$ arms, the bias can be at most $K|S||A|$. 

Now, for the $\ge N - K|S||A|$ arms for which we do take an action as recommended by $\malpha^{(t-1)}_{t-1}$, the next state of the arm follows a categorical distribution (multinoulli distribution) described in Lemma~\ref{lm:multinomial}. Formally, for an arm $j$ that is in cluster $i_j$ and state $s_j$ at time $t-1$, and for which we play the action $a_j$ as recommended by $\malpha^{(t-1)}_{t-1}$, the state of this arm at time $t$ is a random variable that follows a categorical distribution with parameter $P_{t-1,i_j}(\cdot|s_j, a_j)$. Now, as we have $\le N_i$ independent samples of such type for cluster $i$, using Lemma~\ref{lm:multinomial}, we can bound the expected value of the $\ell_1$-distance for a cluster $i$ as:
\[
    \Exp[|| \mmu^{(t)}_{t,i} - \mmu^{(t-1)}_{t,i} ||_1] \le  \sqrt{|S| N_i} + |S||A|.
\]
Putting this together for all clusters, we get
\begin{multline*}
    \Exp[|| \mmu^{(t)}_{t} - \mmu^{(t-1)}_{t} ||_1] = \sum_i \Exp[|| \mmu^{(t)}_{t,i} - \mmu^{(t-1)}_{t,i} ||_1] \\
    \le K|S||A| +  \sum_i \sqrt{|S| N_i} \le  K|S||A| +  \sqrt{K|S|N},
\end{multline*}
where the last inequality is because $\sum_i \sqrt{N_i}$ given $\sum_i N_i = N$ is maximized when $N_i = N/K$ for all $i$. Notice that we could have directly applied Lemma~\ref{lm:multinomial} to $|| \mmu^{(t)}_{t} - \mmu^{(t-1)}_{t} ||_1$ by assuming that a state is a tuple $(i,s) \in [K] \times S$, which should give the same result. In a similar fashion, we can provide a high-probability bound for the $\ell_1$-distance as:
\begin{align*}
    \Prob\left[|| \mmu^{(t)}_{t} - \mmu^{(t-1)}_{t} ||_1 \ge K|S||A| + \sqrt{2 \log(2) K|S|N + 2 N \log(\frac{1}{\delta}) } \right]& \\
    \le \delta&,
\end{align*}
for all $0 < \delta < 1$.
\end{proof}

\begin{proof}[Proof of Lemma~\ref{lm:expected2lp}]
Let $\rmu_t$ and $\ralpha_t$ denote the (random) state and action at time $t$ when using the policy $\rpi$. We show that $(\Exp[\rmu_t], \Exp[\ralpha_t])_{t \in [T]}$ is a feasible solution of the mean-field LP with objective value $\Exp[\rV_{1:T}^{\rpi}(\rmu_1)]$. So, the optimal objective value of the LP, $\mV_{1:T}(\mmu_1)$, is at least as good $\Exp[\rV_{1:T}^{\rpi}(\rmu_1)]$.

Notice that $\ralpha_t$ satisfies the budget constraint (equation \eqref{eq:mflp:4}) almost surely, therefore, $\Exp[\ralpha_t]$ also satisfies the budget constraint. Also notice that the reward at time $t$, $R_t(\ralpha_t)$, is a linear function of $\ralpha_t$, so by linearity of expectation, $\Exp[R_t(\ralpha_t)] = R_t(\Exp[\ralpha_t])$. Further, given $\ralpha_t$, $\rmu_{t+1}$ follows the transition probabilities $P_t$, and we can write its expected value $\Exp[\rmu_{t+1}]$ using $\Exp[\ralpha_t]$ as given in equation \eqref{eq:mflp:2}. 
\end{proof}

\subsection{Sub-Optimality Upper Bound for Infinite Horizon}
\label{sec:mf:expecationIF:app}

\begin{proof}[Proof of Theorem~\ref{thm:expecationIF}]
For an infinite horizon planning problem, we solve the problem for a truncated finite horizon. Let us denote this truncated finite horizon by $T$.
The proof follows the same sequence of steps as the proof for Theorem~\ref{thm:expecationF}, with adjustments for the discount factor. First, we prove a Lipschitz condition, similar to Lemma~\ref{lm:lipsMF}.
\begin{lemma}\label{lm:lipsMFIF}
Let $\mu_t$ and $\mu_t'$ be two arbitrary states at time $t$ with $|| \mu_t - \mu_t' ||_1 = \delta$. Then the optimal objective value of the mean-field policy for time steps $t$ to $T$ starting from $\mu_t$ and $\mu_t'$ is bounded by 
\[
    |V_{t:T}(\mu_t) - V_{t:T}(\mu_t')| \le \frac{ (\gamma^{t-1} - \gamma^{T}) R_{max} \delta}{2 (1 - \gamma)}. 
\]
\end{lemma}
The proof for Lemma~\ref{lm:lipsMFIF} is a trivial extension of the proof of Lemma~\ref{lm:lipsMF}, all arguments follow exactly. The only difference is that in Lemma~\ref{lm:lipsMF} we had an error of $R_{max} \delta / 2$ for a given time step $\tau$, now we have a discounted value of this error $\gamma^{\tau-1} R_{max} \delta / 2$. Summing these error terms for each time step, we get our result. The same idea extends Lemma~\ref{lm:rounding} into Lemma~\ref{lm:roundingIF} and Lemma~\ref{lm:real2lp} to Lemma~\ref{lm:real2lpIF} given below.

\begin{lemma}\label{lm:roundingIF}
$\rV_{1:T}^{\mpi}(\rmu_1)$ is bounded as $| \rV_{1:T}^{\mpi}(\rmu_1) - \sum_{t \in [T]} R_t(\malpha^{(t)}_t)| \le K|S||A|R_{max} \frac{1 - \gamma^T}{1 - \gamma}$.
\end{lemma}

\begin{lemma}\label{lm:real2lpIF}
We have the following bound: 
$| \sum_{t \in [T]} R_t(\malpha^{(t)}_t) - \mV_{1:T}(\mmu_1) | \le \sum_{t \in [2:T]} \frac{(\gamma^{t-1} - \gamma^{T}) R_{max} || \mmu^{(t)}_t - \mmu^{(t-1)}_t ||_1 }{2 (1 - \gamma)}$.
\end{lemma}

Further notice that Lemma~\ref{lm:mubound} and Lemma~\ref{lm:expected2lp} hold directly as they had nothing to do with whether we were using discounted or non-discounted reward. 

\paragraph{Mean-Field Expected Reward}
As we did for Theorem~\ref{thm:expecationF}, we combine the lemmas above to get the following bound on the error incurred for the truncated time horizon $[1:T]$
\begin{align} \label{eq:thm:expecationIF:1}
    &\rV_{1:T}(\rmu_1) - \Exp[\rV^{\mpi}_{1:T}(\rmu_1)] \le  K|S||A|R_{max} \frac{1 - \gamma^T}{1 - \gamma} \nonumber \\
    & \qquad +  \sum_{t \in [2:T]} \frac{(\gamma^{t-1} - \gamma^{T}) R_{max} (K|S||A| +  \sqrt{K|S|N}) }{2 (1 - \gamma)} \nonumber \\
    &\le \frac{R_{max}((2-\gamma)K|S||A| + \gamma  \sqrt{K|S|N)})}{2(1-\gamma)^2} \nonumber \\
    & \qquad- \frac{ \gamma^T (T-1) R_{max} (K|S||A| +  \sqrt{K|S|N})}{2(1-\gamma)}.
\end{align}
If we truncate the horizon at $T$, the maximum total loss in reward for $t > T$ is bounded above by $N R_{max} \gamma^T / (1- \gamma)$. If we select the truncation $T \ge \frac{2 \sqrt{N}}{ \sqrt{K|S|}}  + 1$, then
\begin{align*}
    &T \ge \frac{2 \sqrt{N}}{ \sqrt{K|S|}} + 1 \Longleftrightarrow (T-1)  \sqrt{K|S|N} \ge 2 N \\
    &\implies (T-1) (K|S||A| +  \sqrt{K|S|N})\ge 2 N \\
    &\Longleftrightarrow \frac{ \gamma^T (T-1) R_{max} (K|S||A| +  \sqrt{K|S|N})}{2(1-\gamma)} \ge \frac{\gamma^T R_{max} N}{1- \gamma}.
\end{align*}
Putting this in \eqref{eq:thm:expecationIF:1} we get
\begin{equation*}
    \rV_{1:\infty}(\rmu_1) \le \Exp[\rV^{\mpi}_{1:\infty}(\rmu_1)] + \frac{R_{max}((2-\gamma)K|S||A| + \gamma  \sqrt{K|S|N)})}{2(1-\gamma)^2}.
\end{equation*}

\paragraph{High-Probability Bound for Mean-Field Reward}
The proof is combines the ideas used in the proof for expected reward above and the proof for the high-probability bound in Theorem~\ref{thm:expecationF}. Let $T$ be the truncation time. Let the probability of failure for time $t \in [T]$ be $\delta/T$. By union bound, the total probability of failure is $\delta$. 

Following the steps used to derive \eqref{eq:thm:expecationIF:1}, but using the high-probability bound of Lemma~\ref{lm:mubound}, with probability at least $(1-\delta)$, for every $0 < \delta < 1$, we have 
\begin{align} \label{eq:thm:expecationIF:2}
    &\rV_{1:T}(\rmu_1) - \rV^{\mpi}_{1:T}(\rmu_1) \nonumber \\
    &\le \frac{R_{max}((2-\gamma)K|S||A| + \gamma  \sqrt{2 \log(2) K|S|N + 2 N \log(T/\delta) })}{2(1-\gamma)^2} \nonumber \\
    & \qquad - \frac{ \gamma^T (T-1) R_{max} (K|S||A| +  \sqrt{2 \log(2) K|S|N + 2 N \log(T/\delta) }}{2(1-\gamma)}.
\end{align}
As the total loss in reward for $t > T$ is bounded above by $N R_{max} \gamma^T / (1- \gamma)$. If we select the truncation $T = \left\lceil \frac{\sqrt{2N}}{ \sqrt{\log(2)K|S| + \log(1/\delta)}} \right\rceil  + 1$, then
\begin{align*}
    &T \ge \frac{\sqrt{2N}}{ \sqrt{\log(2)K|S| + \log(1/\delta)}} + 1 \\
    &\Longleftrightarrow (T-1)  \sqrt{2\log(2)K|S|N + 2N\log(T/\delta)} \ge 2 N \\
    &\implies (T-1) (K|S||A| +  \sqrt{2\log(2)K|S|N + 2N\log(T/\delta)})\ge 2 N \\
    &\Longleftrightarrow \frac{ \gamma^T (T-1) R_{max} (K|S||A| +  \sqrt{2 \log(2) K|S|N + 2 N \log(T/\delta) }}{2(1-\gamma)} \\
    & \qquad \qquad \qquad \qquad \qquad \qquad \qquad \qquad \qquad \qquad \ge \frac{\gamma^T R_{max} N}{1- \gamma}.
\end{align*}
Further, as $T = \left\lceil \frac{\sqrt{2N}}{ \sqrt{\log(2)K|S| + \log(1/\delta)}} \right\rceil + 1 \le \sqrt{2N} \le N$ for large enough $N$, which implies $\log(T) \le \log(N)$.\footnote{If $T$ is selected to be larger than $N$, which is allowed as per the theorem statement, we can do a similar analysis by bounding the loss for time $\le N$, and separately (almost surely) bounding for $t > N$ as we are doing for $t > T$ here.}
Putting this in equation \eqref{eq:thm:expecationIF:2} we get
\begin{multline*}
    \rV_{1:\infty}(\rmu_1) - \rV^{\mpi}_{1:\infty}(\rmu_1) \\
    \le \frac{R_{max}((2-\gamma)K|S||A| + \gamma  \sqrt{2 \log(2) K|S|N + 2N \log(N/\delta) })}{2(1-\gamma)^2},
\end{multline*}
with probability at least $1 - \delta$ for every $0 < \delta < 1$. 
\end{proof}

\subsection{Sub-Optimality Lower Bound}\label{sec:mf:lowerbound:app}
\begin{theorem}\label{thm:mf:lowerbound}
There are RMAB instances where the mean-field algorithm has an expected reward that is $\Omega(T\sqrt{N})$ less that the optimal expected reward.
\end{theorem}
\begin{proof} 


We shall prove a lower bound of $\Omega(T\sqrt{N})$ for the sub-optimality of the mean-field policy using the following example.
\begin{example}\label{ex:mflowerbound1}
Let there be only one cluster with $8$ states as shown in Figure~\ref{fig:mflowerbound1}. The solid arrows in Figure~\ref{fig:mflowerbound1} denote the \textit{active} actions and the dotted arrows the \textit{passive} actions; the transitions not shown lead to state $s_5$. All the transitions are deterministic except the transitions from $s_2$ to $s_4$ and $s_2$ to $s_5$ for the passive action and $s_3$ to $s_4$ and $s_3$ to $s_5$ for the active action, each happens with probability $1/2$. The reward for state $s_6$ is $1$ and for $s_8$ is $1-\epsilon$ for an appropriately set $\epsilon > 0$, all other states have reward of $0$. We have $N = 3n$ arms, $2n$ of which start from state $s_1$ and $n$ of which start from $s_7$. 
\end{example}

\begin{figure}[htbp]
\includegraphics[width=0.8\columnwidth]{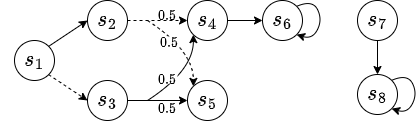}
\caption{Example~\ref{ex:mflowerbound1} - Mean-Field Sub-Optimality}\label{fig:mflowerbound1}
\end{figure}

For the instance described in Example~\ref{ex:mflowerbound1}, for $T \ge 4$ and $\epsilon = \frac{2+(\delta/n)}{T-1}$, where $\delta > 0$ and $\delta \rightarrow 0$, we shall show that the mean-field policy (\mfp) has an expected loss of reward of $\ge (T-3) \sqrt{n/(6\pi)} - \delta$.

Let us first focus on the case $T=4$ and prove a bound of $\Omega(\sqrt{N})$. Later we extend the result to $T \ge 4$.

\paragraph{\textbf{Proof for} $T=4$} By solving the mean-field LP, we can see that it allocates all its budget to the arms that start in state $s_1$, because it overestimates the expected reward for doing so. In particular, the \mfp estimates the following outcome:
\begin{itemize}
    \item $t=1$. \\
    \textbf{State.} There are $2n$ arms in $s_1$ and $n$ arms in $s_7$, which is the initial state. \\
    \textbf{Action.} The \mfp plays active action for $n$ arms (equal to the total budget) out of the $2n$ arms in state $s_1$.
    
    \item $t=2$. \\
    \textbf{State.} The $n$ arms in $s_1$ for which we played the active action at time $t=1$ end up in state $s_2$, while the $n$ arms in $s_1$ for which we played the passive action end up in state $s_3$. All the arms in state $s_7$ end up in the dropout state $s_5$ because we played the passive action for them at $t=1$. \\
    \textbf{Action.} Play the active action for the $n$ arms in state $s_3$.
    
    \item $t=3$. \\
    \textbf{State.} As per \mfp's estimate: half of the arms in state $s_2$ end up in state $s_4$ and half in state $s_5$; similarly, half of the arms in state $s_3$ end up in state $s_4$ and half in state $s_5$. So, we have $0.5n + 0.5n = n$ arms in state $s_4$ and $0.5n + 0.5n + n$ arms in state $s$ ($n$ of these are from the $s_7$ that reached this dropout state at $t=2$). \\
    \textbf{Action.} Play the active action for the $n$ arms in $s_4$.
    
    \item $t=4$. \\
    \textbf{State.} $n$ arms reach $s_6$ from $s_4$. These arms give a total reward of $n$. All other arms are in the dropout state $s_5$. \\
    \textbf{Action.} Not relevant because the reward is homogeneous across actions and the time horizon is $T=4$ (so next state does not matter).
\end{itemize}
As described above, the \mfp estimates that it can get a reward of $n$ from spending its budget on the arms that start from state $s_1$. On the other hand, if it were to select active actions for all the arms in state $s_7$, it would have received a reward of $n(1-\epsilon)$ for time steps $t = 2, 3, 4$, so a total of $3n(1-\epsilon) = 3n(1-\frac{2+(\delta/n)}{3}) = n - \delta$. As $\delta > 0$, \mfp prefers spending its budget on the arms in $s_1$. Using the same argument, we can also see that \mfp would prefer to allocate all its budget to the arms in $s_1$ rather than distributing it among arms in $s_1$ and $s_7$.

Now, let us estimate the expect reward of \mfp. Although \mfp makes a deterministic estimate that $n$ arms reach state $s_6$ at $t = 4$, in reality the number of arms that reach state $s_6$ is a random variable. Let us denote this random variable by $Z$. Notice that over the transition from $t=2$ to $t=3$, $2n$ arms ($n$ from $s_2$ and $n$ from $s_3$) move to state either $s_4$ or $s_5$ with probability $1/2$ each. Let $X$ be the random variable that denotes the number of arms that reach $s_4$. This is a binomial random variable with $2n$ unbiased coin tosses. Now, from $t=3$ to $t=4$, at most $n$ arms can move from $s_4$ to $s_6$ because of the budget constraint of $n$. So, $Z = \min(n, X)$ is the number of arms that reach $s_6$ and give a reward of $1$ each. Let us compute the expected value of $Z$ as follows:
\begin{align*}
    \Exp[Z] &= \frac{1}{2^{2n}} \sum_{i = 0}^{2n} \binom{2n}{i} \min(n, i) \\
    &= \frac{1}{2^{2n}} \left( \sum_{i = 0}^{n} i \binom{2n}{i} + \sum_{i = n+1}^{2n} n \binom{2n}{i} \right) \\
    &= \frac{1}{2^{2n}} \left( 2n \sum_{i = 1}^{n} \binom{2n-1}{i-1} + \frac{n}{2} \left(-\binom{2n}{n} + \sum_{i = 0}^{2n} \binom{2n}{i} \right) \right) \\
    &\le \frac{1}{2^{2n}} \left( 2n \sum_{i = 0}^{n-1} \binom{2n-1}{i} + \frac{n}{2} \left(-\frac{2^{2n}}{\sqrt{\pi (n + 0.5)}}+ 2^{2n} \right) \right) \\
    &= \frac{1}{2^{2n}} \left( n \sum_{i = 0}^{2n-1} \binom{2n-1}{i} + \frac{n}{2} \left(-\frac{2^{2n}}{\sqrt{\pi (n + 0.5)}}+ 2^{2n} \right) \right) \\
    &= \frac{1}{2^{2n}} \left( n 2^{2n} - \frac{n2^{2n}}{2 \sqrt{\pi (n + 0.5)}} \right) \\
    %
    %
    &= n - \frac{n}{2 \sqrt{\pi (n + 0.5)}} \le n - \sqrt{\frac{n}{6\pi}}.
\end{align*}
In the first inequality, we are upper bounding the central binomial coefficient $\binom{2n}{n}$ by $\frac{2^{2n}}{\sqrt{\pi (n + 0.5)}}$.\footnote{\url{https://en.wikipedia.org/wiki/Central_binomial_coefficient}} In the last inequality, we are using the fact that $n \ge 1$.

As discussed earlier, the policy that spends its budget on arms that start in $s_7$ gets a reward of $n-\delta$ (with probability $1$). So, the loss of expected reward for \mfp is $\ge n - \delta - (n - \sqrt{n/(6\pi)})$, which is $\Omega(\sqrt{n}) = \Omega(\sqrt{N})$, as total number of arms is $N = 3n = \Theta(n)$.

\paragraph{\textbf{Extension to} $T \ge 4$} The extension to $T \ge 4$ is easy. At $t = 4$, we have shown that $\le n - \sqrt{n/(6\pi)}$ arms reach state $s_6$ (in expectation) and all other arms go to the dropout state $s_5$. For $t \ge 4$, it can be observed that \mfp will continue playing the active action for these arms and get a total expected reward of $\le (T-3) (n - \sqrt{n/(6\pi)})$. On the other hand, a policy that allocates budget to arms that start in $s_7$ gets a reward of $(T-1)n(1-\epsilon) = (T-3)n - \delta$. So, the loss in expected reward is $\ge (T-3) \sqrt{n/(6\pi)} - \delta = \Omega(T\sqrt{N})$.
\end{proof}

\section{Alternate RMAB Policy from the Mean-Field LP}\label{sec:mf:alternate:app}
Here, we describe an alternate mean-field policy that solves the mean-field LP only once at $t=1$, in contrast to the algorithm described in Section~\ref{sec:mf:policy} that solves the LP every time step.

\subsection{Algorithm}
Given the initial state $\rmu_1 = \mmu_1$, solve the mean-field LP given in Section~\ref{sec:mf:lp} to get the estimated mean-field states $\mmu_t$ and actions $\malpha_t$ for $t \in [T]$. For $t = 1, 2, \ldots, T$, do the following:
\begin{enumerate}
    \item Given the real-life state $\rmu_t$ at time $t$, play the action $a$ for  $\left\lfloor \malpha_{t,i}(s,a) \frac{\min(\mmu_{t,i}(s), \rmu_{t,i}(s))}{\mmu_{t,i}(s)} \right\rfloor$ arms in cluster $i$ and state $s$. For the remaining arms in cluster $i$ and state $s$, which is equal to $\rmu_{t,i}(s) - \sum_a \left\lfloor \malpha_{t,i}(s,a) \frac{\min(\mmu_{t,i}(s), \rmu_{t,i}(s))}{\mmu_{t,i}(s)} \right\rfloor$ play the zero cost action.
    \item After playing the action at time $t$, we get the realization of the real-life state at time $t+1$, $\rmu_{t+1}$.
\end{enumerate}

\subsection{Analysis (Sketch)} 
It is easy to check that the policy described above is a valid policy because (i) we define a valid action for all arms at each time step and (ii) the actions played at each time step have a total cost of at most the cost of  $\malpha_t$, which is $\le B_t$, as required.

To see the near-optimality of this policy, we can use the results that we proved in Section~\ref{sec:mf:analysis}. The main high-level steps of the proof are as follows:
\begin{itemize}
    \item Given $\rmu_1 = \mmu_1$, we can approximate $\Exp[|| \rmu_2 - \mmu_2 ||_1] \approx  \sqrt{K|S|N}$ using Lemma~\ref{lm:mubound}, ignoring constant factors and lower order terms.
    \item As we play the action recommended by the mean-field LP for about (ignoring rounding error) $\sum_{i,s} \min(\mmu_{t,i}(s), \rmu_{t,i}(s)) \approx N - || \rmu_2 - \mmu_2 ||_1/2$ arms, using ideas similar to the ones used in the proof of Lemma~\ref{lm:mubound}, we can show that $\Exp[|| \rmu_3 - \mmu_3 ||_1] \approx 2 \sqrt{K|S|N}$.
    \item Repeating the idea, for all $t$, we get $\Exp[|| \rmu_t - \mmu_t ||_1] \approx (t-1) \sqrt{K|S|N}$, ignoring constant factors and lower-order terms.
    \item Putting it together, we get that the reward of the policy described above is $\approx \mV_{1:T}(\mmu_1) - T^2 \sqrt{K|S|N} \ge \rV_{1:T}(\rmu_1) - T^2 \sqrt{K|S|N}$, where the last inequality is from Lemma~\ref{lm:expected2lp}.
\end{itemize}
\section{Additional Experimental Results}\label{sec:exp:app}
In this section, we provide additional results for the experiments in Section~\ref{sec:maternity} and Section~\ref{sec:tuberculosis}. We also describe and provide results for another simulation environment by \cite{killian2021beyond}.

\subsection{Mobile Healthcare for Maternity Health (Additional Plots)}\label{sec:maternity:app}

The real-life data collected by \cite{mate2022field} contains the weekly behavior of $96158$ beneficiaries for about $12$ weeks. For each beneficiary, the data provides the action (whether called or not-called by a healthcare worker) and the state (whether engaged or not-engaged) for every week. As done in \cite{mate2022field}, we make a Markov assumption and estimate the transition probabilities from this data, and then do clustering on passive transition probabilities (because the data is very sparse and most beneficiaries do not have any samples for active transitions). 

The implementation is in Python. The mean-field LP used PuLP with CBC optimizer. Other algorithms like finite and infinite horizon Whittle were written using Numpy. For each set of hyperparameters, we averaged the rewards over $40$ runs with random start states and transitions. The runs were executed in parallel on a machine with $48$ cores and $128$ GB RAM. 


Figures~\ref{fig:armman:app:t50} and~\ref{fig:armman:app:t100} vary the number of clusters from $40$ to $100$ on the $x$-axis, plot the difference in reward from the random policy on the $y$-axis, while setting $T=50$ and $T=100$, respectively. Both these figures show that the \mfp algorithm achieves a higher reward than both the finite and the infinite versions of the Whittle index algorithm.
Figure~\ref{fig:armman:app:timevnoclusters} plots the runtime of the \mfp algorithm and the finite Whittle index algorithm with respect to varying number of clusters, keeping time horizon fixed at $T=100$. It plots the number of clusters on the $x$-axis and the time in seconds on the $y$-axis. 
Figures~\ref{fig:armman:app:timec20} to~\ref{fig:armman:app:timec100} vary the number of clusters from $20$ to $100$ and plot the average runtime of the Whittle algorithm and the mean-field algorithm with respect to the length of time horizon. 
While \mfp is slower than Whittle index, it scales linearly as the number of clusters or the time horizon grows, demonstrating its suitability for real-world practical settings.

\begin{figure*}[ht]
  \centering
  \subfloat[Reward vs Number of Clusters ($T = 50$)]{
  \includegraphics[width=0.4\textwidth]{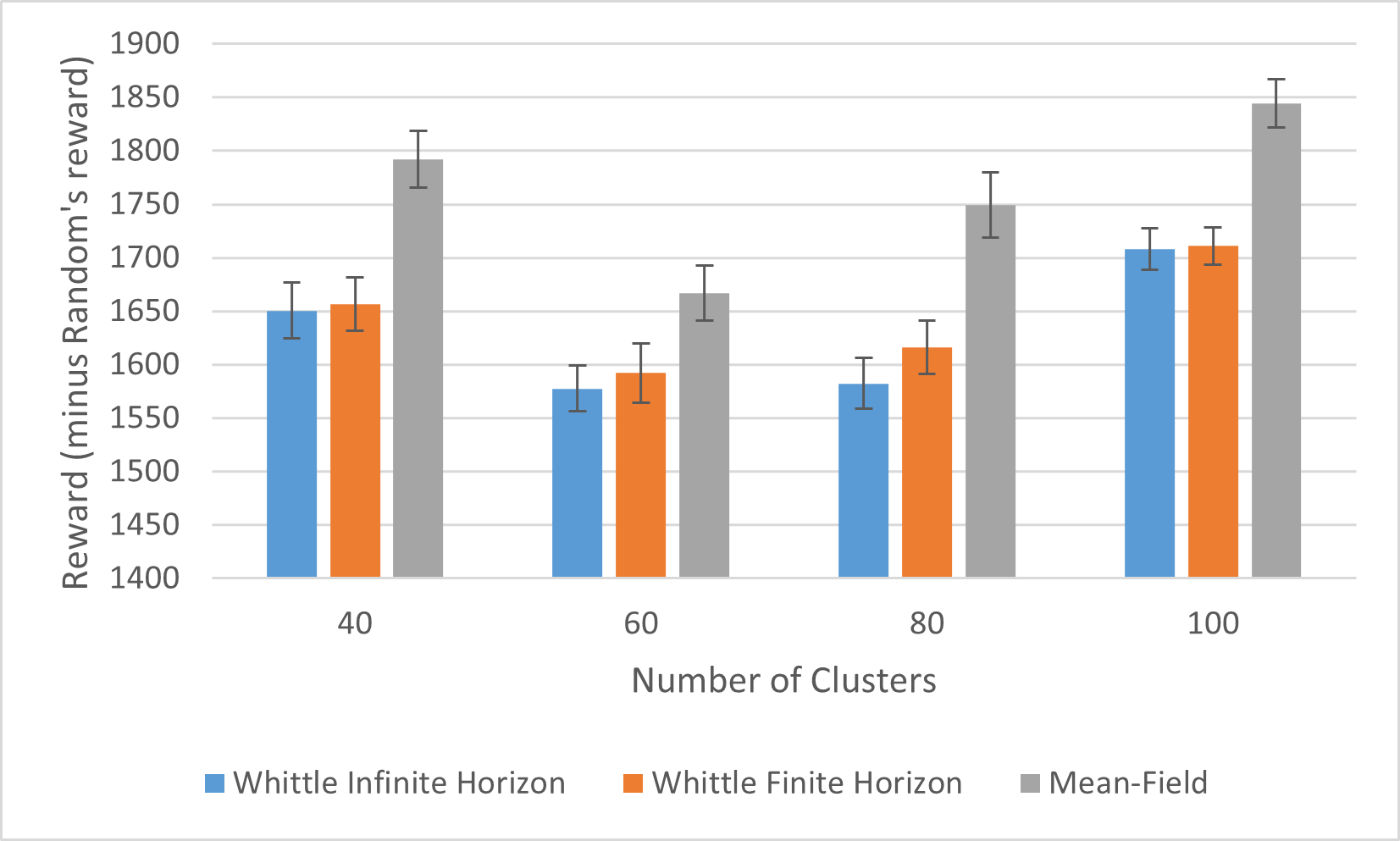}
  \label{fig:armman:app:t50}
  }
  \subfloat[Reward vs Number of Clusters ($T = 100$)]{
  \includegraphics[width=0.4\textwidth]{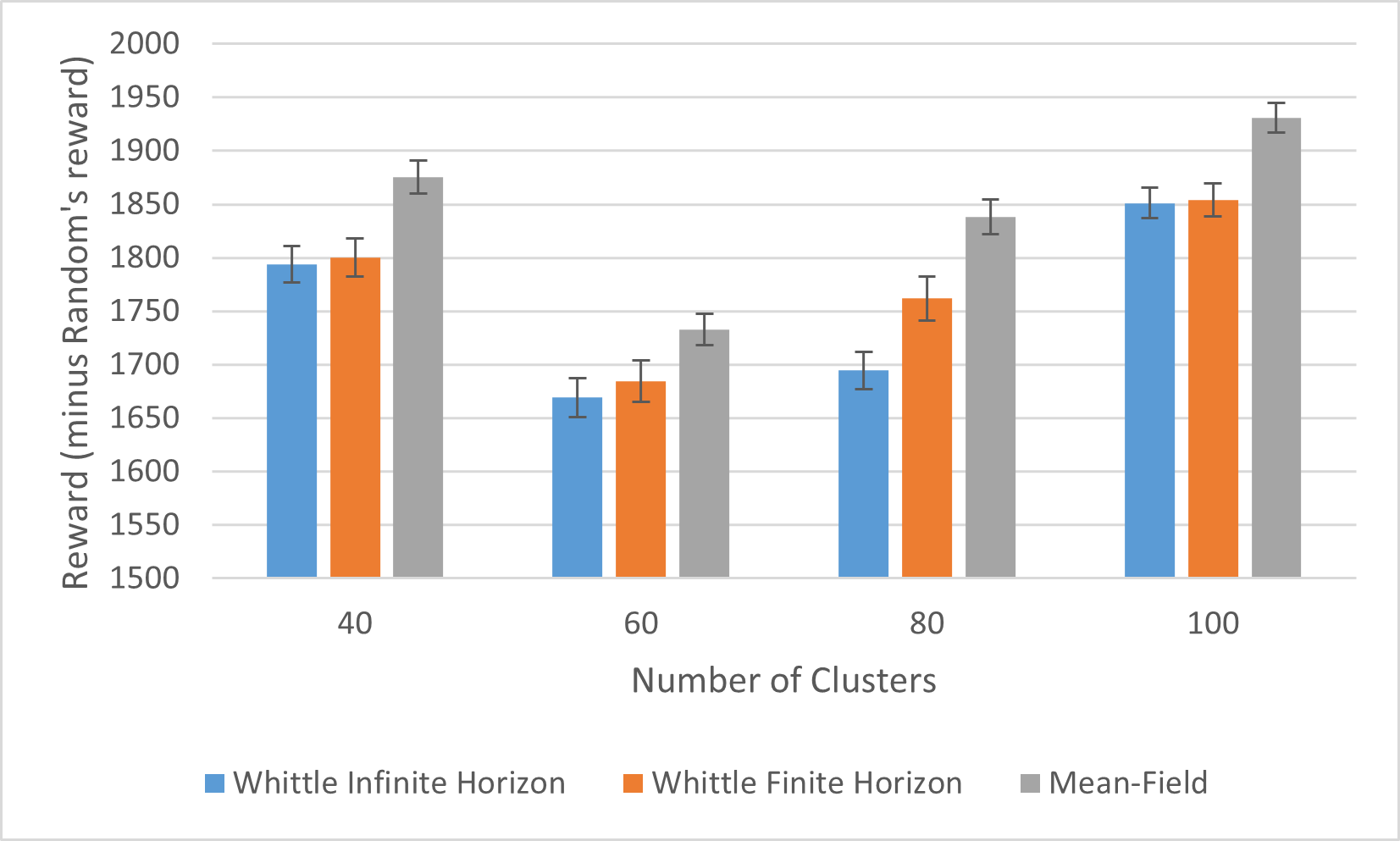}
  \label{fig:armman:app:t100}
  } 
  \\ 
  \subfloat[Runtime vs Number of Clusters ($T = 100$)]{
  \includegraphics[width=0.4\textwidth]{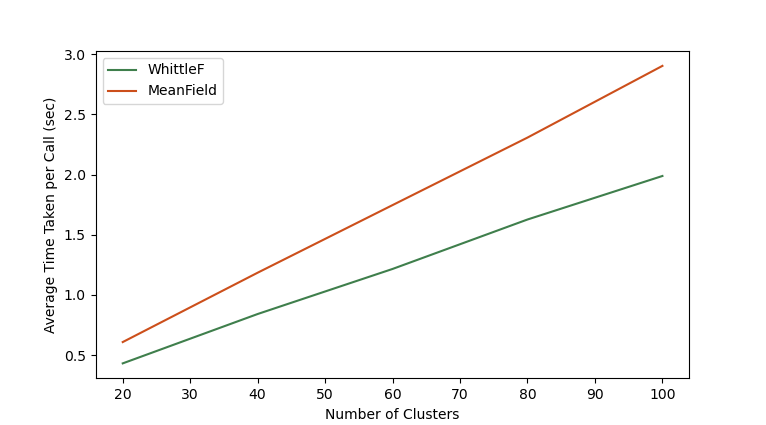}
  \label{fig:armman:app:timevnoclusters}
  }
  \subfloat[Runtime vs Time Horizon ($K = 20$ )]{
  \includegraphics[width=0.4\textwidth]{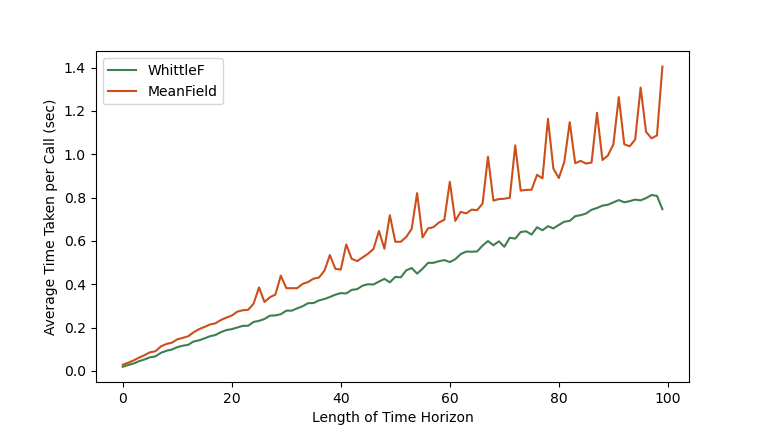}
  \label{fig:armman:app:timec20}
  }
  \\
  \subfloat[Runtime vs Time Horizon ($K = 40$)]{
  \includegraphics[width=0.4\textwidth]{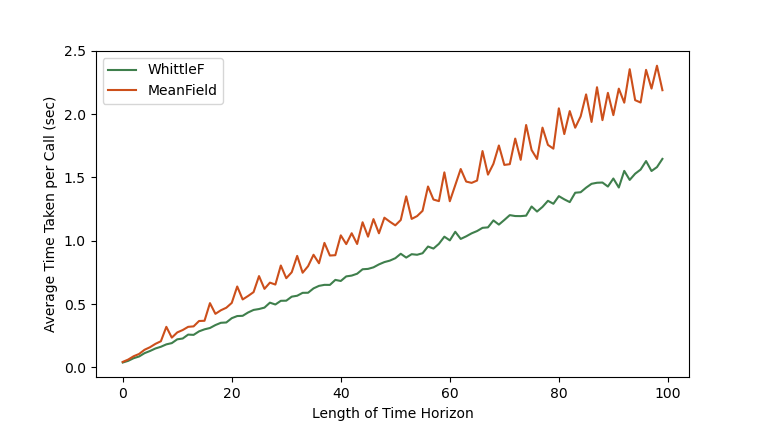}
  \label{fig:armman:app:timec40}
  }
  \subfloat[Runtime vs Time Horizon ($K = 60$)]{
  \includegraphics[width=0.4\textwidth]{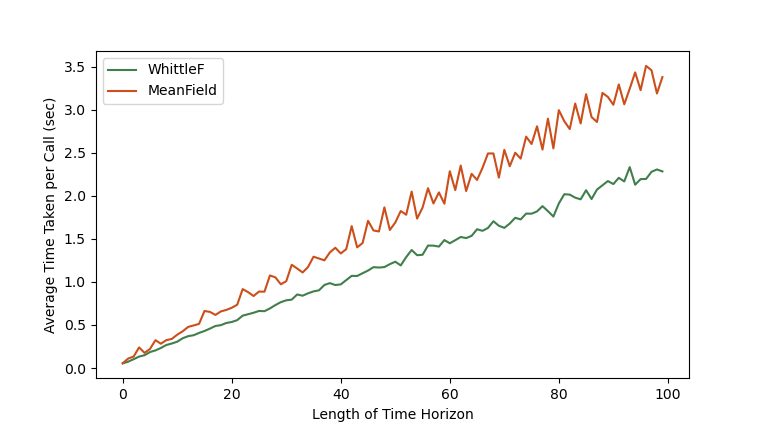}
    \label{fig:armman:app:timec60}
  }
  \\
  \subfloat[Runtime vs Time Horizon ($K = 80$)]{
  \includegraphics[width=0.4\textwidth]{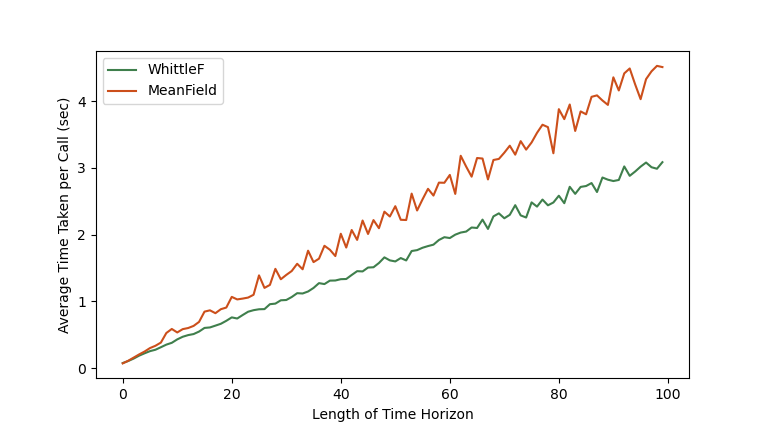}
  \label{fig:armman:app:timec80}
  }
  \subfloat[Runtime vs Time Horizon ($K = 100$)]{
  \includegraphics[width=0.4\textwidth]{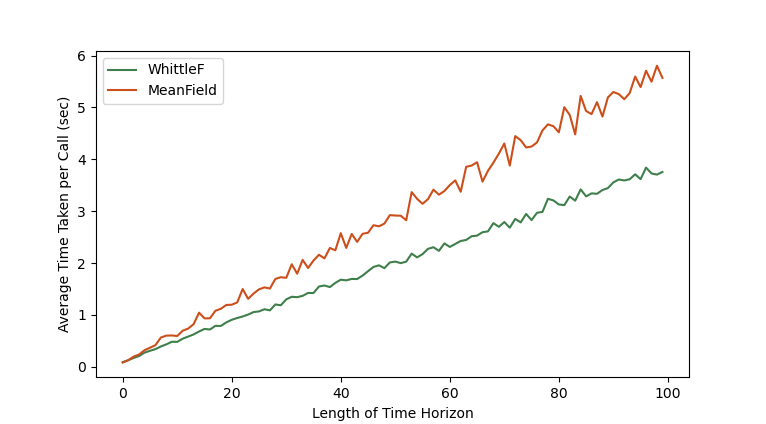}
  \label{fig:armman:app:timec100}
  }
\caption{Additional experiments on the ARMMAN domain (plots best seen in color)}
\label{fig:armman:app:extraexp}
\end{figure*}

\subsection{Tuberculosis Healthcare (Additional Plots)}\label{sec:tuberculosis:app}
We built upon the simulation environment developed by \cite{killian2021beyond}.\footnote{\url{https://github.com/killian-34/MAMARB-Lagrange-Policies}} We used the implementation for Hawkins and other algorithms in \cite{killian2021beyond} using Gurobi optimizer. We added a mean-field implementation in Gurobi to this simulation environment. All our results are averaged over $25$ runs, and each set of $25$ runs was executed in parallel on a machine with $48$ cores and $128$ GB RAM. 

We now provide additional results comparing mean field with the algorithms in~\citet{killian2021beyond} varying the number of arms. Figures~\ref{fig:tb:app:arms50} to~\ref{fig:tb:app:arms1000} vary the number of arms from $50$ to $1000$. The $y$-axis in these plots shows the discounted sum of rewards, whereas the $x$-axis shows the different policies.
These results show that the mean-field reward is as good as the Hawkins reward (and the reward of other policies). 


\begin{figure*}[ht]
  \centering
  \subfloat[Reward vs Policy ($N = 50$ arms)]{
  \includegraphics[width=0.4\textwidth]{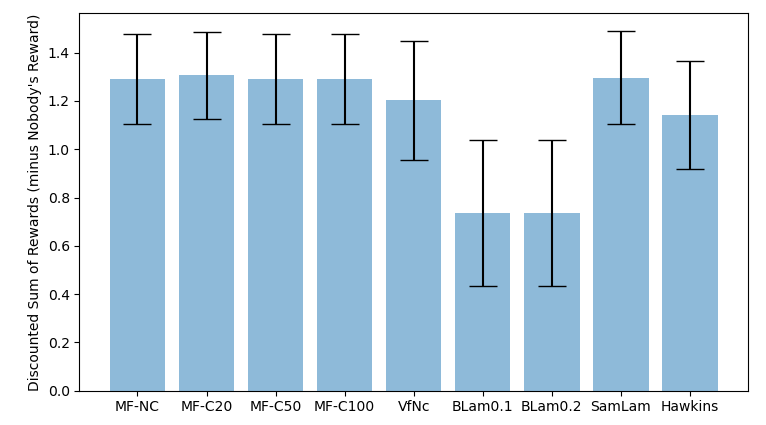}
  \label{fig:tb:app:arms50}
  }
  \subfloat[Reward vs Policy ($N = 100$ arms)]{
  \includegraphics[width=0.4\textwidth]{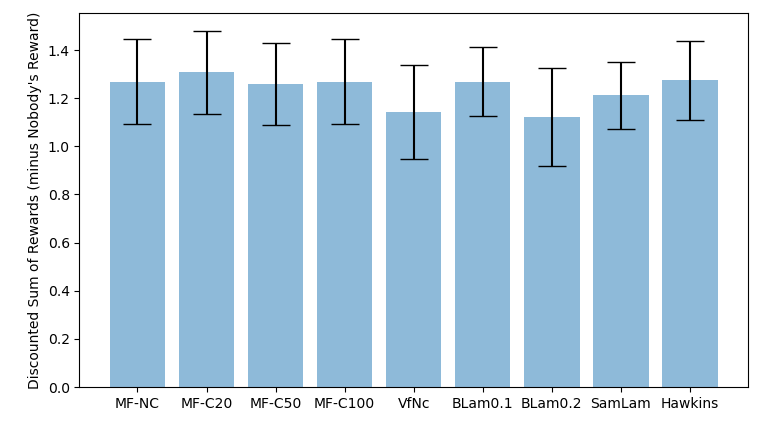}
  \label{fig:tb:app:arms100}
  }
  \\
  \subfloat[Reward vs Policy ($N = 200$ arms)]{
  \includegraphics[width=0.4\textwidth]{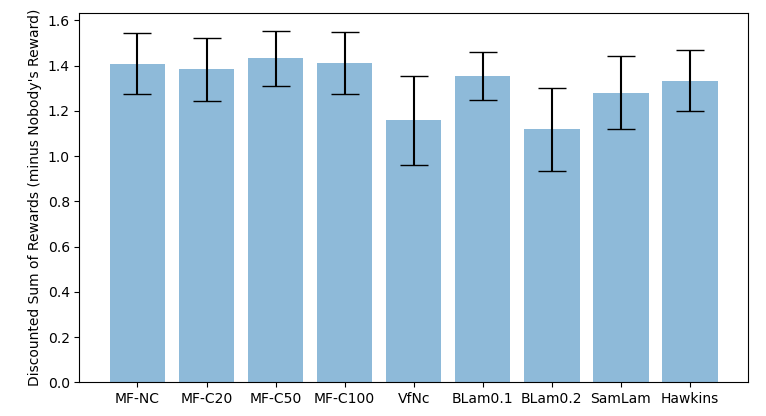}
  \label{fig:tb:app:arms200}
  }
  \subfloat[Reward vs Policy ($N = 500$ arms)]{
  \includegraphics[width=0.4\textwidth]{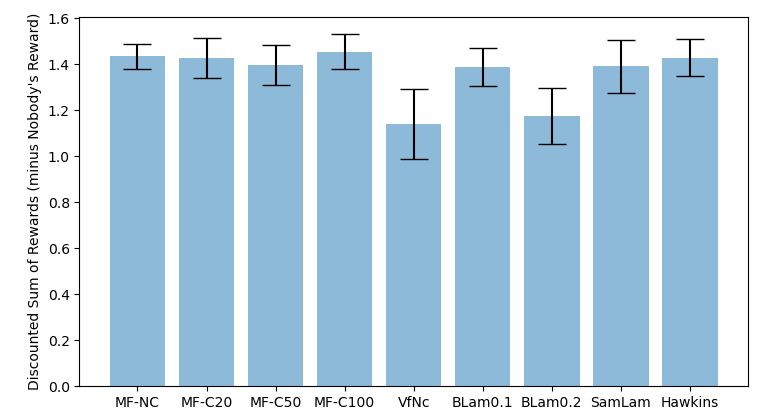}
  \label{fig:tb:app:arms500}
  }
  \\
  \subfloat[Reward vs Policy ($N = 1000$ arms)]{
  \includegraphics[width=0.4\textwidth]{plots/tuberculosis/1000.png}
  \label{fig:tb:app:arms1000}
  }
  \subfloat[Runtime vs Number of Arms (Mean-Field without Clustering and Hawkins*10)]{
  \includegraphics[width=0.4\textwidth]{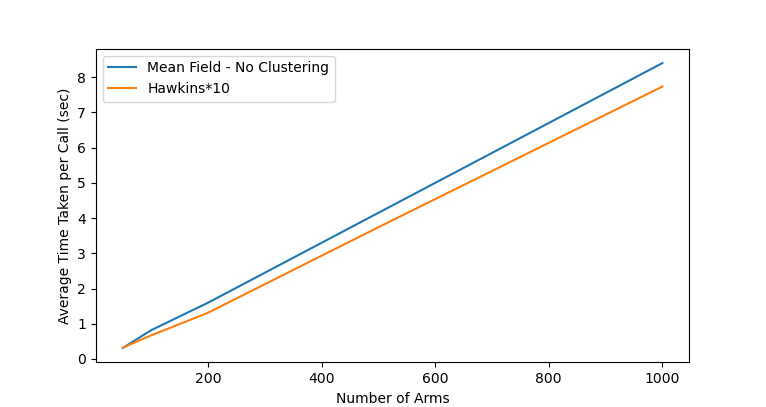}
  \label{fig:tb:app:time1}
  }
  \\
  \subfloat[Runtime vs Number of Arms (Mean-Field with Clustering and Hawkins*)]{
  \includegraphics[width=0.4\textwidth]{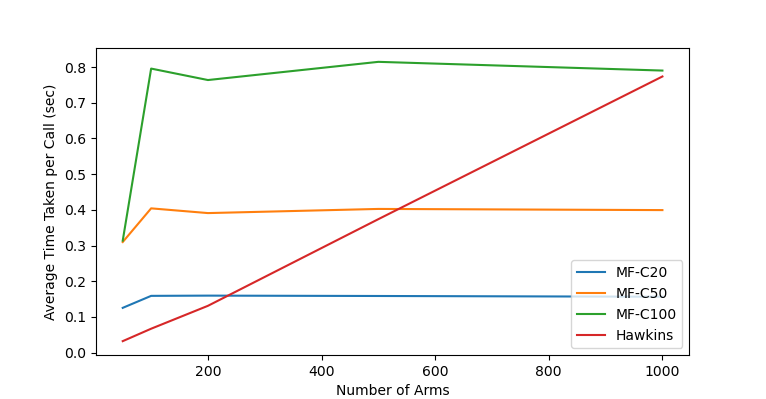}
  \label{fig:tb:app:time2}
  }

\caption{Tuberculosis Experiment: Additional Results (plots best seen in color)}
\label{fig:tb:app}
\end{figure*}





On the other hand, Figures~\ref{fig:tb:app:time1} and~\ref{fig:tb:app:time2} plot the time taken by the mean-field algorithms (with and without clustering) and the Hawkins algorithm versus the number of arms. Here, the $y$-axis shows the time in seconds, whereas the $x$-axis shows the number of arms. Without clustering, the time taken by the mean-field algorithm scales linearly with the number of arms. With clustering, the time taken remains (approximately) a constant as the number of arms increases. (In our experiments, the other algorithms in \cite{killian2021beyond} take the following proportion of time in comparison to Hawkins: VfNc $\approx 3\%$, SampleLam $\approx 20\%$, BLam0.1 $\approx 35\%$, Blam0.2 $\approx 32\%$.)






\subsection{Additional Experimental Domain}
\label{sec:gre:app}
\citet{killian2021beyond} propose another simulation environment which is similar to the examples we saw in Section~\ref{sec:fail_case}. There are three types of agents: (1) \textit{Greedy}: Must take increasingly expensive actions to collect increasingly high reward. Once the required action is not taken, the agent never produces reward again. This is modeled with a single chain of states, each with unit-increasing reward, reachable only by an action with unit-increasing cost. Failure to take the next action leads to a dead state. (2) \textit{Reliable}: Must take the cheapest non-zero action every round to achieve reward $1$. If the arm is not played for any round, it never produces reward again. This is modeled with a simple $2$-state chain, in which the final state recurs with the proper action, otherwise it goes to a dead state. (3) \textit{Easy}: Always gives reward of $1$ regardless of action. We make the proportion of (1) and (2) equal and set the budget so that the cheapest non-zero action can be played for at most all of (1) or (2) (or some mix), but not more. 

Notice that the instance described above is a generalized version of the examples in  Section~\ref{sec:fail_case}. Clearly, with large enough horizon, the optimal policy is to always play the Reliable agents since committing to the Greedy agents will eventually leave the planner only collecting reward from the Easy agents. But, as proved in Section~\ref{sec:fail_case}, the Whittle index policy performs poorly in such settings. The Hawkins algorithm is the generalization of Whittle index policy to multiple actions. Figure~\ref{fig:gre:1} shows that Hawkins performs very poorly compared to the mean-field policy for an instance with $10$ actions. Figure~\ref{fig:gre:2} shows a similar result for Hawkins with $2$ actions, equivalently the Whittle index policy. The time horizon is $40$, the discount factor is $0.95$, and the number of arms is $50$.

\begin{figure*}[ht]
  \centering
  \subfloat[$10$ Actions]{
  \includegraphics[width=0.4\textwidth]{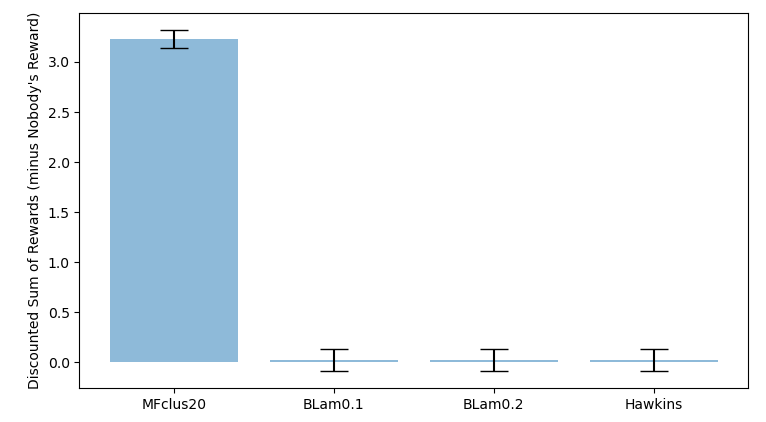}
  \label{fig:gre:1}
  }
  \subfloat[$2$ Actions (Hawkins $\equiv$ Whittle)]{
  \includegraphics[width=0.4\textwidth]{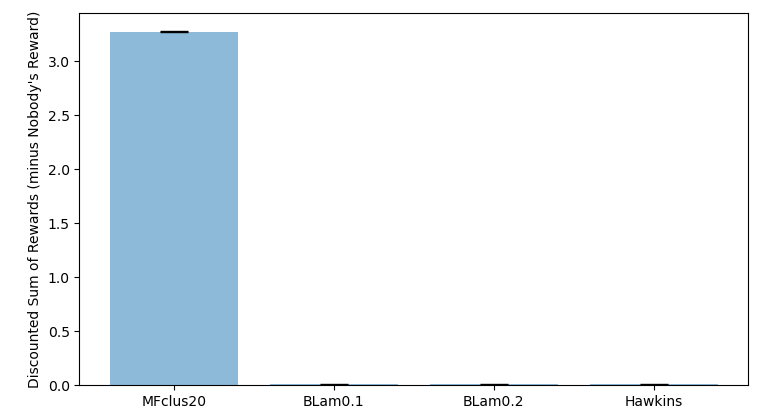}
  \label{fig:gre:2}
  }

\caption{Greedy/Reliable/Easy Experiment from~\cite{killian2021beyond} - Reward vs Policy}
\label{fig:gre}
\end{figure*}

\end{document}